\documentclass[twocloumn, 10pt]{IEEEtran}
\IEEEoverridecommandlockouts 

\usepackage{verbatim} 
\usepackage{hyperref} 
\usepackage{graphicx, caption}
\usepackage{amssymb}
\usepackage{amsmath}
\usepackage{amsthm}
\theoremstyle{definition}
\usepackage{mathtools}
\usepackage{cite}
\usepackage{stfloats}
\usepackage{epstopdf}
\usepackage{subfig}
\usepackage{psfrag}
\usepackage{pstool}
\usepackage[mathscr]{euscript}
\usepackage{acronym}  
\usepackage{booktabs} 
\usepackage[table]{xcolor}
\usepackage[labelsep=period]{caption} 

\usepackage{color}
\usepackage{dsfont}
\usepackage{bbm}

\acrodef{AoI}{age of information}
\acrodef{PAoI}{peak \ac{AoI}}
\acrodef{BS}{base station}
\acrodef{HLF}{Hyperledger Fabric}
\acrodef{HPPP} {homogeneous Poisson point process}
\acrodef{BeMN}{blockchain-enabled monitoring networks}
\acrodef{MVCC}{multi-version concurrency check}
\acrodef{STP}{successful transmission probability}
\acrodef{FCFS}{first-come-first-serve}
\acrodef{PDF}{probability density function}
\acrodef{i.i.d.}{independent and identically distributed}
\acrodef{KS}{Kolmogorov-Smirnov}
\acrodef{CDF}{cumulative distribution function}
\acrodef{CCDF}{complementary cumulative distribution function}
\acrodef{PoW}{proof-of-work}
\acrodef{SD}{standard deviation}
\usepackage{color}
\usepackage{dsfont}
\usepackage{bbm}

\newtheorem{definition}{Definition}
\newtheorem{theorem}{Theorem}
\newtheorem{lemma}{Lemma}
\newtheorem{cor}{Corollary}
\newtheorem{proposition}{Proposition}
\newtheorem{example}{Example}
\newtheorem{remark}{Remark}
\newtheorem{property}{Property}
\newtheorem{algorithm}{Algorithm}
\newtheorem{assumption}{Assumption}

\newcommand{\limit}[2]{\lim_{#1 \to\infty} #2}
\newcommand{\E}{\mathbb{E}}
\renewcommand{\P}{\mathbb{P}}
\newcommand{\integral}[2]{\int_{#1}^{#2}}
\newcommand{\summation}[2]{\sum_{#1}^{#2}}
\newcommand{\aoivio}{$P_v$}
\newcommand{\Aoivio}{P_v}
\newcommand{\avgaoi}{$\bar{\Delta}$}
\newcommand{\Avgaoi}{\bar{\Delta}}
\newcommand{\paoivio}{$P_{\text{p}^v}$}
\newcommand{\Paoivio}{P_{\text{p}^v}}

\newcommand{\TotLat}[1][k]{T_{\text{tot},#1}}

\newcommand{\blocksize}{$B$}
\newcommand{\Blocksize}{B}
\newcommand{\timeout}{$T$}
\newcommand{\Timeout}{T}

\newcommand{\Teff}{T_{\text{eff}, k}}
\newcommand{\Tint}{T_\text{int}}

\newcommand{\TargetSTP}{\zeta}

\newcommand{\TargetAoI}{v}
\newcommand{\PAoI}{\Delta_{\text{p}, k}}
\newcommand{\Location}{x_{\text{o}}}
\newcommand{\transLat}{$T_{\text{tx}, \Location}$}
\newcommand{\TransLat}{T_{\text{tx}, \Location } }
\newcommand{\TransLatk}{T_{\text{tx},\Location, k}}

\newcommand{\Remain}{T_{v \hspace{-0.2mm}, \hspace{-0.1mm} {\Location}}}

\allowdisplaybreaks 
\begin{document}
	\newcommand{\paperTitle}{Title}



\title{ Ensuring Data Freshness for\\ 
	Blockchain-enabled Monitoring Networks}

\author{
	\vspace{0.2cm}
	
	\thanks{
		M.\ Kim, S.\ Lee, C.\ Park, and J.\ Lee are with    
		the Department of Information and Communication Engineering, Daegu Gyeongbuk Institute of Science and Technology, Daegu, South Korea, 42988      
		(e-mail:\{\texttt{kms0603}, \texttt{seuho2003}, \texttt{pcw0311}, \texttt{jmnlee}\}@dgist.ac.kr). 
		
		W.\ Saad is with the Wireless@VT Group, Bradley Department of Electrical and Computer Engineering, Virginia Tech, Blacksburg, VA 24061 USA (e-mail: \texttt{walids@vt.edu}).
		
		The material in
		this paper will be presented, in part, at the IEEE International Conference on Communications,
		Montreal, Canada, Jun. 2021 \cite{MKSLCPJL:21}
		
		This research was supported, in part, by the U.S. Office of Naval Research (ONR) under  MURI Grant N00014-19-1-2621.
	}
	
\IEEEauthorblockN{
	Minsu~Kim, Sungho~Lee, Chanwon~Park,
	Jemin~Lee, \textit{Member, IEEE}, and
	Walid Saad, \textit{Fellow, IEEE}
}
%

%

}

\maketitle 

%

%

%

%

\acresetall
	\begin{abstract}
	The age of information (AoI) is a recently proposed metric for quantifying data freshness in real-time status monitoring systems where timeliness is of importance. In this paper, the problem of characterizing and controlling the AoI is studied in the context of blockchain-enabled monitoring networks (BeMN). In BeMN, status updates from sources are transmitted and recorded in a blockchain. To investigate the statistical characteristics of the AoI in BeMN, the transmission latency and the consensus latency are first rigorously modeled. 
	Then, the average AoI, the AoI violation probability, and the peak AoI violation probability are derived in a closed form so as to quantify the performance of BeMN.
	Furthermore, a simplified form is derived for the AoI violation probability, and it is shown that this quantity can capture the upper or lower bound of the actual AoI violation probability. Simulation results show that each BeMN parameters (i.e., target successful transmission probability, block size, and timeout) can have conflicting effects on the AoI-related performance. Subsequently, design insights are provided to maintain the freshness of the status data in BeMN. Then, experimental results with a real Hyperledger Fabric platform further validate the accuracy of our modeling and analysis.
	\end{abstract} 
	

\section{introduction}	
The emergence of blockchains has ushered in a new breed of decentralized data management platforms that have been adopted in a broad range of real-time wireless Internet of Things (IoT) monitoring applications ranging from healthcare systems to pollution detection and autonomous factories \cite{ON:18}. 
In such real-time monitoring systems, data integrity is one of the important requirements in order to prevent unintended or malicious changes in the data. 
Through integration with blockchains, real-time monitoring systems can maintain data integrity in a distributed manner without the need for a central authority.
Although blockchains can provide a secure data management platform for monitoring systems by ensuring data integrity, they cannot guarantee data freshness. The use of outdated data when making system-wide decisions, such as raising an alarm due to a high temperature or level of pollution, can lead to incorrect outputs. Those outputs can further jeopardize the operation of a whole system. Therefore, it is of importance to maintain fresh data for blockchain-enabled monitoring systems to prevent undesired outputs.

To quantify the degree of data freshness, the notion of \ac{AoI} has been proposed in \cite{SaRoMa:12}. The \ac{AoI} is defined as the elapsed time from the generation of the latest received status update. Several variants of the \ac{AoI} recently appeared, including the so-called \ac{PAoI}, which measures the largest staleness of information \cite{MaMaAn:16}. In blockchain-enabled monitoring systems, status updates are recorded in distributed ledgers. Hence, the \ac{AoI} can be used as a suitable metric to quantify data freshness in the ledgers of a blockchain. However, in order to optimize the \ac{AoI} and maintain data freshness in a blockchain, several challenges must be addressed such as accounting for the transaction processing latency to update ledgers.

There has been a number of prior works that looked at the measurement and analysis of transaction processing latency in blockchain platforms \cite{LeDo:19, LaPe:20, AlCa:20, MABS:19}.
The authors in \cite{LeDo:19} and \cite{LaPe:20} implemented a permissioned blockchain integrated IoT platform as a proof of concept for monitoring networks, and they evaluated the performance of the proposed networks in terms of latency. The work in \cite{AlCa:20} investigated the effects of block generation frequency on the end-to-end latency in Ethereum-based blockchain platforms under cellular and Wi-Fi networks.
The authors in \cite{MABS:19} studied the effects of the number of network hops and replica nodes on the end-to-end latency in Byzantine fault tolerance-based blockchain platforms. 
Although the latency for a status update in blockchain platforms is studied and evaluated in \cite{LeDo:19, LaPe:20, AlCa:20, MABS:19}, data freshness was not characterized. For maintaining data freshness in blockchain-enabled networks, one must consider jointly the latency in processing transactions, the transaction generation frequency, and the communication latency.

Recently, in \cite{ArAb:20}, the authors analyzed data freshness, in terms of average \ac{AoI}, for a public blockchain and the IOTA platform. However, the distribution of the \ac{AoI} is not studied although it is essential to show the percentage of the time guaranteeing a certain target \ac{AoI}. Moreover, the latency to process a transaction is simply modeled to follow an exponential distribution, without the validation in a real blockchain platform. We also acknowledge that a number of works \cite{DoSa:20, BaYu:20, AlMo:20, MAHD:19, XSZN:19, BoWa:19} looked at different aspects related to the analysis and optimization of the \ac{AoI} and \ac{PAoI} in monitoring networks. Moreover, the work in \cite{ChWa:20} quantified the worst-case achievable \ac{PAoI} for an augmented reality system over future wireless networks. However, these prior works do not investigate the synergies between blockchains and AoI, especially for a \emph{permissioned blockchain}. Note that a permissioned blockchain can be more suitable for the data management in IoT platforms than a public blockchain as it avoids the use of an intensive consensus protocol (e.g., mining) under a strict membership rule.

The main contribution of this paper is, thus, a novel framework to analyze the data freshness of monitoring networks that use a permissioned blockchain for data management. We call such networks \ac{BeMN}, and it consists of sources, \acp{BS}, and an \ac{HLF} \cite{Git}, which is one of the most widely used permissioned blockchain platform \cite{KuSa:19}. In \ac{BeMN}, sources monitor physical phenomena, and each source transmits its monitored data to an associated \ac{BS}, which is connected to the \ac{HLF} network. Following the consensus process in the \ac{HLF} network, the monitored data is stored in distributed ledgers, and this newly monitored data is used as a status update of the source. To measure the freshness of data in \ac{BeMN}, we analyze the distribution of the \ac{PAoI} and the average \ac{AoI} of \ac{BeMN} by considering both the transmission latency and the consensus latency. We then obtain the \ac{AoI} violation probability, which captures the probability that the \ac{AoI} exceeds a target \ac{AoI}. Furthermore, we  explore the effects of communication and \ac{HLF} parameters on the average \ac{AoI}, the \ac{PAoI} violation probability, and the \ac{AoI} violation probability. Our main contributions can thus be summarized as follows. 

\begin{itemize}
	\item We characterize, in a closed-form,  the statistical characteristics of the \ac{AoI} for \ac{BeMN} including the average \ac{AoI}, the \ac{PAoI} violation probability, and the \ac{AoI} violation
	probability by considering the consensus latency in an \ac{HLF} network as well as the transmission latency. 
	
	\item We explore the impacts of the communication parameter 
	(i.e., target \ac{STP}) and the \ac{HLF} parameters (i.e., block size and timeout)
	on the \ac{AoI} in \ac{BeMN}. In particular, simulation results show that each \ac{BeMN} parameter can have conflicting effects on the \ac{AoI}-related performance. Hence,
	a higher target \ac{STP} or a larger value for the \ac{HLF} parameters  will not always guarantee a lower \ac{AoI}.
		
	\item We implement a real \ac{HLF} platform (v1.3) for measuring the consensus latency. We then validate the accuracy of our modeling and analysis through the implemented \ac{HLF} platform. Experimental results also corroborate the impacts of the \ac{BeMN} parameters on the \ac{AoI}-related performance.
\end{itemize}

The rest of this paper is organized as follows. Section~\ref{sec: Hyperledger Fabric} introduces the overall transaction flow in an \ac{HLF} and the associated parameters. 
Section~\ref{sec: BeMN} describes the \ac{HLF} blockchain-enabled monitoring network and models the consensus latency. 
Section~\ref{sec: AoI} derives the distribution of the \ac{PAoI}, the average \ac{AoI}, and the \ac{AoI} violation probability in \ac{BeMN}.
Section~\ref{Sec: Numerical Results} provides the validation of the analytical results and effects of the parameters of \ac{BeMN} on the \ac{AoI} violation probability.
Finally, conclusions are drawn Section~\ref{sec: Conclusion}.

Notation: An overview of our notation is shown
in Table \ref{Table:notation}.

\begin{table}
	\caption{Summary of our key notations.} \label{Table:notation}
	\begin{center}
		\rowcolors{2}
		{cyan!15!}{}
		\renewcommand{\arraystretch}{1.3}
		\begin{tabular}{c p{6cm} }
			\hline 
			{\bf Notation} & {\hspace{2.5cm}}{\bf Definition}
			\\
			\midrule
			\hline
			$\Blocksize$ & Block size of an \ac{HLF} network \\ \addlinespace
			$\Timeout$ & Timeout of an \ac{HLF} network \\ \addlinespace
			$\alpha$ & Shape parameter of the Gamma distribution\\ \addlinespace
			$\beta$ & Rate parameter of the Gamma distribution\\ \addlinespace 
			$\lambda_s$ 	& Spatial density of sources \\ \addlinespace
			$\lambda$		& Spatial density of \ac{BS}s \\ \addlinespace
			$P$ 		& Transmission power of sources	\\ \addlinespace
			$n$ & Pathloss exponent\\ \addlinespace
			$\bar{\epsilon}$ & Maximum target data rate\\ \addlinespace			
			$\TargetSTP$ & Target STP \\ \addlinespace
			$D$ & Packet size\\ \addlinespace
			$\Tint$ & Inter-generation time of two consecutive packets that successfully arrive at the BS\\ \addlinespace
			$\rho_s$ 	& Generation rate of packets at a source \\ \addlinespace
			$\rho$ & Generation rate of packets that successfully arrive at the \ac{BS}\\ \addlinespace
			$\TotLat[k]$ 		& Total latency of packet $k$ to update status \\ \addlinespace
			$X_k$	& Consensus latency of packet $k$ \\ \addlinespace
			$\TransLat$ 	& Transmission latency \\ \addlinespace
			$\Teff$		& Inter-generation time of two effective packets $k-1$ and $k$ \\ \addlinespace
			$G_k$ 		& Generation instant of effective packet $k$ \\ \addlinespace
			$A_k$ 		& Arrival instant of effective packet $k$ at the BS \\ \addlinespace
			$U_k$		& Update instant of effective packet $k$ \\ \addlinespace
			$\PAoI$ 		& Peak \ac{AoI} of effective packet $k$ \\ \addlinespace
			$\TargetAoI$ & Target \ac{AoI} \\ \addlinespace
			$T_k^v$ 	& Time duration of the \ac{AoI} being larger than  $\TargetAoI$ between $U_k$ and $U_{k-1}$ \\ \addlinespace
			$\Remain$ & Time difference between $\TargetAoI$ and $\TransLat$ \\ \addlinespace
			$\Avgaoi$ 	& Average \ac{AoI} \\ \addlinespace
			$\Aoivio$ 	& \ac{AoI} violation probability \\ \addlinespace
			$\Paoivio$ & \ac{PAoI} violation probability\\ \addlinespace			
			\hline 
			
		\end{tabular}
	\end{center}\vspace{-0.63cm}
\end{table}%

\section{Hyperledger Fabric: Preliminaries} \label{sec: Hyperledger Fabric}
In this subsection, we present the overall structure of \ac{HLF} and the components of the consensus process for a status update. We also introduce the \ac{HLF} parameters which affect the performance of \ac{BeMN}. 
\subsection{\ac{HLF} Transaction Flow} \label{subsec: transaction flow}
\ac{HLF} is a permissioned blockchain platform, in which all changes made by transactions are committed to the distributed ledger \cite{Git}. In \ac{HLF}, peer nodes (or peers) hold their own copies of the distributed ledgers. The ledger is a key-value database, which consists of two parts: a blockchain and a world state. In the blockchain, the immutable records of status changes are stored. Meanwhile, the world state is also a database, in which the current value of the status is paired with its own key and current version number.
Hence, all data in the ledger is identified by its own key and version number. Every ledger update starts with the generation of a transaction.
A \emph{transaction} is executed against the specified function to update the stored data in the ledger with each corresponding key.
Here, we assume that the IoT monitoring system's sources generate a transaction to update their status as done in \cite{LeDo:19} and \cite{LaPe:20}. In \ac{HLF}, participants are all identified. Therefore, the costly consensus method used in public blockchains, known as mining, is not necessary. Instead, the consensus process in \ac{HLF} is composed of three phases: endorsement phase, ordering phase, and validation phase are described next and detailed \cite{E1ChChSrChBiMaCh:18} and \cite{SLMKJL:20}. 

\subsubsection{Endorsement Phase}
All transactions for status updates enter the endorsement phase first. 
During this phase, peers simulate a transaction using their ledgers. The peers make sure that they have identical simulation results, which are called endorsements. These endorsements include the updated status and version number of the ledger in the peer. Then, the transaction with the endorsements is transmitted to the ordering node. Note that, although the transaction simulation results are ready, the status is not updated in this phase. 

\subsubsection{Ordering Phase}
The ordering phase is used not only to arrange transactions in a chronological order but also to generate new blocks with the ordered transactions. The ordering nodes continuously include transactions into a new block until it reaches the pre-defined maximum block size. In order to avoid high latency, a timer is prepared with a pre-defined timeout value. If the timer expires, the nodes instantly export the new block, regardless of the current number of transactions in the block. The newly generated block is then delivered to the peers by the ordering nodes.

\subsubsection{Validation Phase}
Next, in the validation phase, the blocks that are delivered to the peers will be validated and the ledger will be updated. This phase consists of two sequential steps: verification and update.
The peers investigate if each transaction in the block is properly endorsed from the endorsement phase. Then, the peers check whether the version numbers in the endorsements are identical to the ones currently stored in their copied ledgers. This verification is also called the \ac{MVCC} verification. Note that the version number increases each time the corresponding status is updated. Hence, if the two version numbers are different, then this signifies that the status has already been updated by the previous transaction before the current one completes the consensus process. If the version numbers are different, then the transaction becomes invalid and ineffective. Finally, the peers update the world state and the blockchain in the ledger.
\vspace{-0.1cm}
\subsection{\ac{HLF} Parameters}
We are interested in two key \ac{HLF} parameters: the block size and the block-generation timeout, which essentially refer to the \ac{HLF} configurations. As introduced in the ordering phase section, those parameters control how long a transaction will wait in the ordering phase, which affects the consensus latency. Hence, for maintaining data freshness in monitoring networks, these parameters should be properly designed to avoid a large waiting time in the ordering phase. 

\subsubsection{Block Size, \blocksize} 
A block size \blocksize\ limits the maximum number of transactions in a block. A newly arrived transaction needs to wait in the ordering phase until the number of transactions in the block reaches \blocksize. Therefore, a larger \blocksize\ will lead to longer waiting times for the transactions because more time is needed to fill up the block.
\subsubsection{Timeout, \timeout}
A timeout \timeout\ is another way to limit the waiting time of a transaction in the ordering phase. The transaction can wait up to the timeout value \timeout\ for other transactions in the ordering phase. The new block can move to the next phase even if the block is not completely full to avoid long latency. As expected, transactions generally need to wait longer as \timeout\ increases.

Given these preliminaries, in the following section, we present our system model based on Section \ref{sec: Hyperledger Fabric}.

	\begin{figure}
	
	\includegraphics[width=1.02\columnwidth]{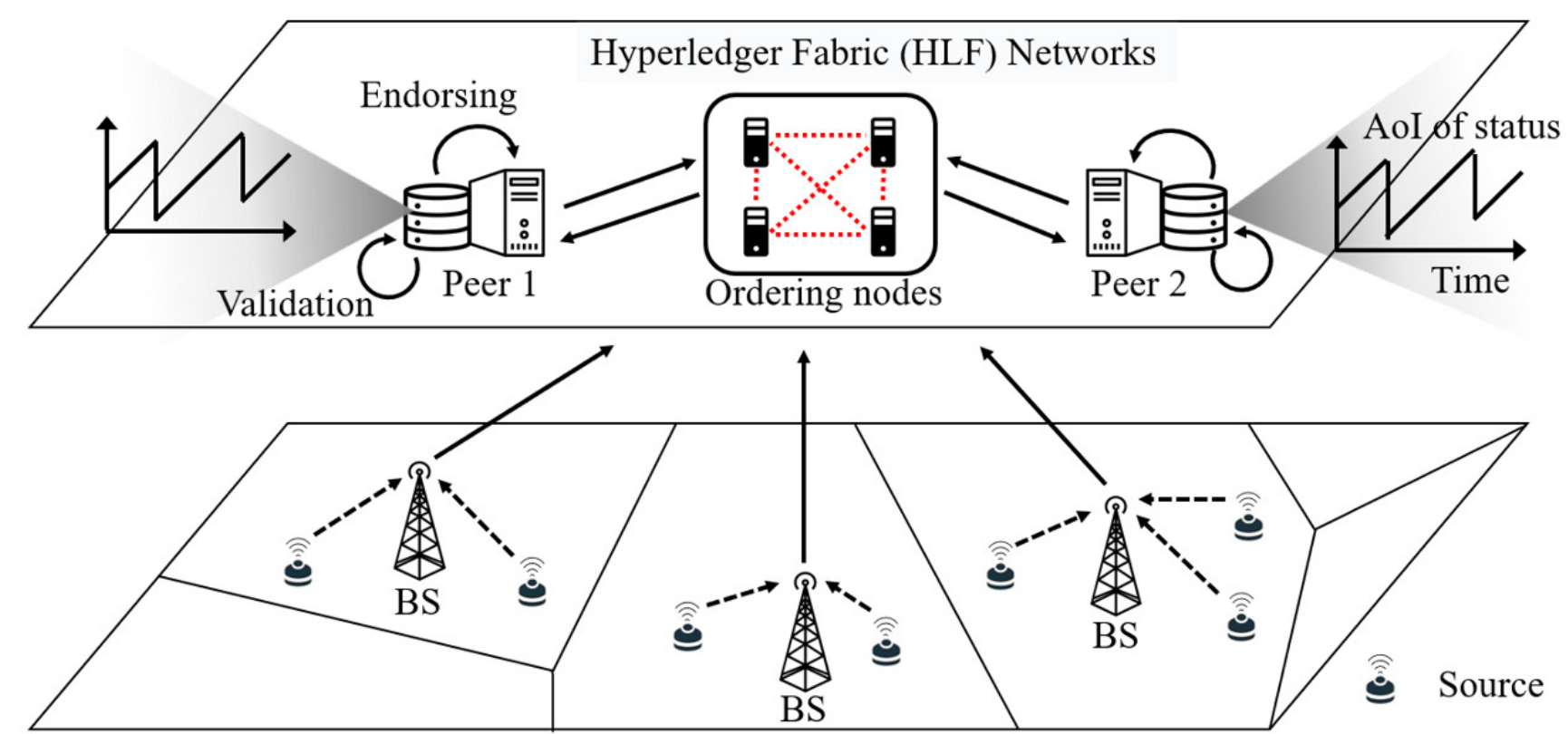}
	\captionsetup {singlelinecheck = false}

	\caption{\ac{BeMN} structure, composed of sources, \acp{BS}, and the \ac{HLF} network.  }
	\label{fig: overall_sysmtem_model}

	\end{figure}

		



\section{Blockchain-Enabled Monitoring Network} \label{sec: BeMN}
\subsection{System Model}
We consider \ac{BeMN} composed of sources, \acp{BS}, and an \ac{HLF} network. In \ac{BeMN}, the sources monitor physical phenomena (e.g., temperature, pollution level) and update the corresponding status stored in the \ac{HLF} as shown in Fig. \ref{fig: overall_sysmtem_model}. We assume that the distribution of the sources follows a \ac{HPPP} $\Phi_s$ with spatial density $\lambda_s$. The source transmits a packet through a wireless uplink channel to the nearest \ac{BS} as done in \cite{ALSJ:09} and \cite{IMTB:14}. We assume that all the sources use the same transmission power $P$. The distribution of the \ac{BS}s also follows an \ac{HPPP} $\Phi_m$ with spatial density $\lambda$. Without loss of generality, we consider a typical source located at the origin and its associated \ac{BS} at $\Location$.
Each channel is allocated to one source only in the cell of any given \ac{BS} to mitigate intra-cell interference between sources belonging to the same cell.

We assume the \ac{BS}s to be connected to the \ac{HLF} network, whereby each status information of the source is stored with its key values. As shown in Fig. \ref{fig: overall_sysmtem_model}, a source monitors a physical phenomenon and generates a packet with newly observed information. In our model, the source generates a packet with an exponentially distributed inter-generation time and rate $\rho_\text{s}$ as done in \cite{ZJGR:03} and \cite{AVAR:17}. The packet is delivered to the \ac{HLF} via a \ac{BS} in the form of a transaction.  Successfully received transactions can update their status information through the consensus process described in Section \ref{subsec: transaction flow}. 
%
%
%

We define the consensus latency as the total time required for the commitment of a transaction, which is the sum of the latencies in each phase. Then, the total latency of packet $k$ will be:
\begin{equation}
\TotLat = X_k + \TransLatk, \label{total_latency}
\end{equation}
where $X_k$ is the consensus latency of packet $k$ and $\TransLatk$ is the transmission latency needed to send a packet from the typical source to its associated \ac{BS}. The consensus latency $\{X_k, k \geq1 \}$ is assumed to be \ac{i.i.d.}.

\subsection{Consensus Latency Modeling}
\label{Consensus Latency Modeling}
We now model the consensus latency in our \ac{BeMN}.
From the empirical results of a constructed \ac{HLF} platform, it is shown in \cite{SLMKJL:21} that the Gamma distribution is reasonable for modeling the consensus latency in the \ac{HLF} platform. Hence, the consensus latency $X_k$ of packet $k$ can be modeled as a Gamma random variable, i.e., $X_k \sim \text{Gamma}(\alpha, \beta)$, whose \ac{PDF} is given by \cite{Thom:58}
\begin{align}
f_{X_k}(x) = \frac{\beta^\alpha}{\Gamma(\alpha)} x^{\alpha - 1} e^{-\beta x}, \label{gamma_pdf}
\end{align}
where $\alpha$ and $\beta$ are the shape and the rate parameters, respectively, and $\Gamma (\cdot)$ is a Gamma function.
To determine the values of $\alpha$ and $\beta$, we use the maximum likelihood estimation \cite{Thom:58}. Specifically, $\alpha$ and $\beta$ can be given by
\begin{align}
&\alpha =
\frac{1}{4 A} 
\left(
1 + \sqrt{1 + \frac{4 A}{3}}
\right), \beta =
\frac{\alpha}{\bar{X}}, \label{average_latency}
\end{align}
where $\bar{X}$ is the mean of consensus latencies, and $A$ is given by 
\vspace{-0.3cm}
\begin{align}
A = \log(\bar{X}) - \summation{i = 1}{N} \log(X_i) /  N
\end{align}
for the sample consensus latency $X_i$ and $N$ number of samples.
We use the \ac{KS} test as done in \cite{ISAFDB:71} and \cite{JSJCNK:05} to quantify the accuracy of the modeling (the accuracy results are given in Section \ref{Sec: Numerical Results}). Note that this modeling of consensus latency is applicable to general \ac{HLF} with version 1.0 or higher.\footnote{\ac{HLF} with version 1.0 or higher includes the \ac{MVCC} verification}

\subsection{Transmission Latency}

Next, we analyze the transmission latency $\TransLatk$ of \ac{BeMN}.
The signal-to-interference-plus-noise ratio (SINR) received by a \ac{BS} located at point $\Location$ under Rayleigh fading channel is 
%
%
\begin{align}
\gamma_{\Location}= \frac{{P h_{\Location} l^{-n}}} {{I_{\Location}} + N_0 W}, \label{SINR_definition}
\end{align}
%
%
%
where  $h_{\Location}$ is the fading channel gain, i.e., $h_{\Location} \sim \exp(1)$, $l$ is the distance between the typical source and the associated \ac{BS}, $n$ is the path loss exponent, $N_0$ is the noise power, and $W$ is the channel bandwidth. In \eqref{SINR_definition}, $I_{x_\text{o}}$ is the inter-cell interference from other sources that use the same uplink frequency band, given by
%
%
\begin{equation}
I_{x_\text{o}} = P \summation{u \in \Psi_u}{} h_{u, x_\text{o}} \|u\|^{-n},
\end{equation}
%
%
where $\Psi_u$ denotes the set of locations of the interfering sources which use the same frequency band with the typical source. We assume that each cell has one source that uses the same uplink frequency band as the typical source. Hence, the density of interfering sources is the same as that of the \ac{BS}s, $\lambda$.  Then, the achievable data rate $R_{\Location}$ is given by
%
\begin{equation}
R_{x_\text{o}} = W \log_2(1 + \gamma_{x_\text{o}}). \label{channel_capacity} 
\end{equation}
%
%
%
We define \ac{STP} $p_c$ as the probability that the achievable data rate $R_{x_\text{o}}$ is greater than or equal to a target rate $\epsilon$, i.e., $p_c = \P \left [R_{x_\text{o}} \geq \epsilon \right ].$ We assume that a packet is transmitted at the maximum target rate $\bar{\epsilon}_{\Location}$ to guarantee $p_c \geq \TargetSTP$, where $\TargetSTP$ is a target \ac{STP}. Hence, $\bar{\epsilon}_{\Location}$ is given by \cite{ChJe:18}
%
%
\begin{align}
\bar{\epsilon}_{\Location} = \max \
\{
\epsilon \ |\ \P[R_{x_\text{o}} \geq \epsilon] \geq \TargetSTP 
\}. \label{Maximum target rate}
\end{align}
Then, $\bar{\epsilon}_{\Location}$ in \eqref{Maximum target rate} can be obtained in the following proposition.
\begin{proposition} \label{prop 1}
The maximum target rate $\bar{\epsilon}_{\Location}$ is the one that satisfies the following equation
\begin{align}
\exp\left( \hspace{-0.5mm} -\frac{l^n}{P} N_0 W \theta({\bar{\epsilon}_{\Location}}) \hspace{-0.5mm} - \frac{2 \lambda \pi^2 l^2 \theta({\bar{\epsilon}_{\Location}})^{2/n}}{n P^{2/n} \sin \left(2\pi/n  \right)}   \right) = \TargetSTP,
\label{transmissionrate fixed r} 
\end{align}
where $\theta({\bar{\epsilon}_{\Location}}) = 2^{\bar{\epsilon}_{\Location}/W} - 1$.
\end{proposition}

\begin{proof}
	See Appendix \eqref{Appendix A}.
\end{proof}

From Proposition 1, we can see that finding a closed-form for $\bar{\epsilon}_{\Location}$ is challenging for a general path loss exponent. However, for $n=4$, $\bar{\epsilon}_{\Location}$ can be derived as:
\begin{align}
\bar{\epsilon}_{\Location} \hspace{-0.5mm} 
=
\hspace{-0.5mm}  W \hspace{-0.5mm}  \log_2 \hspace{-0.9mm}  
\left[ 
\hspace{-0.5mm}  1  \hspace{-1.0mm}  + \hspace{-1.0mm}   \left \{ \hspace{-0.5mm}   \frac{\sqrt{P} (-\pi^2 \lambda  \hspace{-0.5mm}  +  \hspace{-0.5mm}  \sqrt{\pi^4 \lambda^2  \hspace{-0.5mm}  -  \hspace{-0.5mm}  16N_0 W \log \TargetSTP })} {4 N_0 W l^2} \hspace{-0.5mm}   \right\}^2   
\right] \hspace{-0.5mm} . \label{transmissionrate n=4} 
\end{align}

Since the source transmits a packet at the rate $\bar{\epsilon}_{\Location}$, the transmission latency $\TransLat$ can be defined as
\begin{align}
\TransLat \hspace{-0.8mm} &= \hspace{-0.5mm} \frac{{D}}{\bar{\epsilon}_{\Location}} \nonumber \\
&= \hspace{-0.5mm} \frac{{D} \hspace{-0.5mm} \log 2}{W}  \hspace{-0.9mm} \left[ \hspace{-0.5mm}  1  \hspace{-1.0mm}  + \hspace{-1.0mm}   \left\{
\hspace{-0.8mm}
\frac{\sqrt{P} (-\pi^2 \hspace{-0.5mm} \lambda  \hspace{-0.6mm}  +  \hspace{-0.6mm}  \sqrt{\pi^4 \hspace{-0.5mm} \lambda^2  \hspace{-0.6mm}  -  \hspace{-0.6mm}  16N_0 W \hspace{-0.3mm} \log \hspace{-0.2mm} \TargetSTP })}
{4 N_0 W l ^2} \hspace{-0.8mm}   
\right\} ^2 \hspace{-0.5mm}  \right] \hspace{-1mm} , \label{Transmission_latency_closed}
\end{align}
where $D$ [bits] is the packet size. Note that we omit the packet index $k$ in $\TransLat$ since each packet experiences the same transmission latency when the source is at $\Location$. 

\begin{figure}
	\centering
	\begin{center}
		\includegraphics[width=1.03\columnwidth]{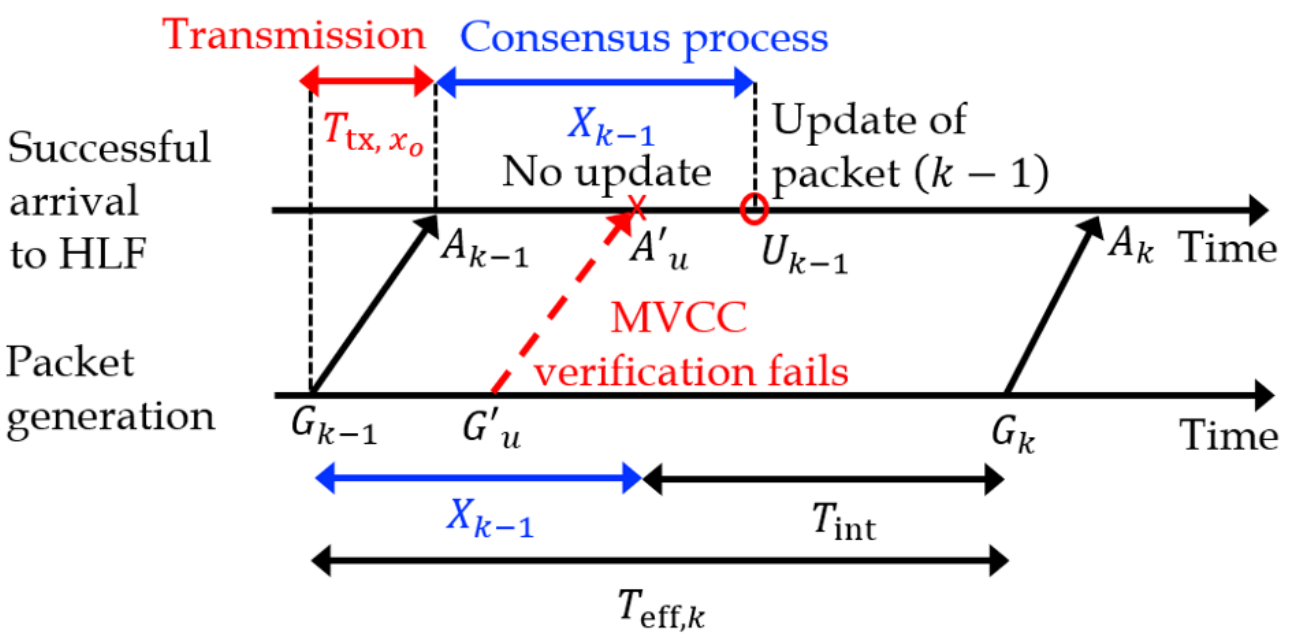}
	\end{center}
	
	\caption{\ac{MVCC} verification failure of the packet generated at $G'_u$.}
	
	\label{fig:MVCC}	
	
\end{figure} 

In \ac{BeMN}, each transmitted packet must go through the consensus process to complete an update. In contrast to conventional databases, in an \ac{HLF}, this consensus latency is not negligible. Hence, we investigate the data freshness using \ac{AoI}-related performance metrics, such as the average \ac{AoI}, the \ac{AoI} violation probability, and the \ac{PAoI} violation probability, by considering both the transmission latency and the consensus latency in the following section.

\section{AoI Analysis of BeMN} \label{sec: AoI}
We now analyze the \ac{AoI} violation probability in our \ac{BeMN}.
As a metric for measuring the data freshness, the \ac{AoI} is defined as the elapsed time since the generation of the latest received packet \cite{SaRoMa:12}. We focus on the \ac{AoI} of the specific status, which is stored with a certain key value in the ledger. As discussed in Section \ref{sec: Hyperledger Fabric}, not every generated packet can make a valid update in \ac{BeMN} because of the \ac{MVCC} verification failure. If the status is updated before the current packet completes its consensus process, this packet becomes invalid and ineffective. We call the packets that make valid updates as \textit{effective packets}. For effective packet $k$, we define $G_k$ as the generation instant at the source, $A_k$ as the arrival instant at the \ac{BS}, and $U_k$ as the update instant at the ledger. As shown in Fig. \ref{fig:MVCC}, packet $k$ can be effective only if its arrival instant $A_{k}$ is after  $U_{k-1}$, which is the update instant of the previous effective packet $(k-1)$. All packets that arrive after $X_{k-1}$ from $G_{k-1}$ become effective packets. Then, the \ac{AoI} at time $t$ can be defined as 
\begin{align}
\Delta(t) = t - \text{max} \{G_k \ | \ U_k \leq t \}.
\end{align}
In Fig. \ref{fig:MVCC}, $G'_u$ and $A'_u$ are, respectively, the generation instant and the arrival instant of invalid packet $u$. We also define the inter-generation time of two consecutive effective packets $\Teff = G_k - G_{k-1}$.

\begin{figure}
	\centering
	\begin{center}
		\includegraphics[width=0.85\columnwidth]{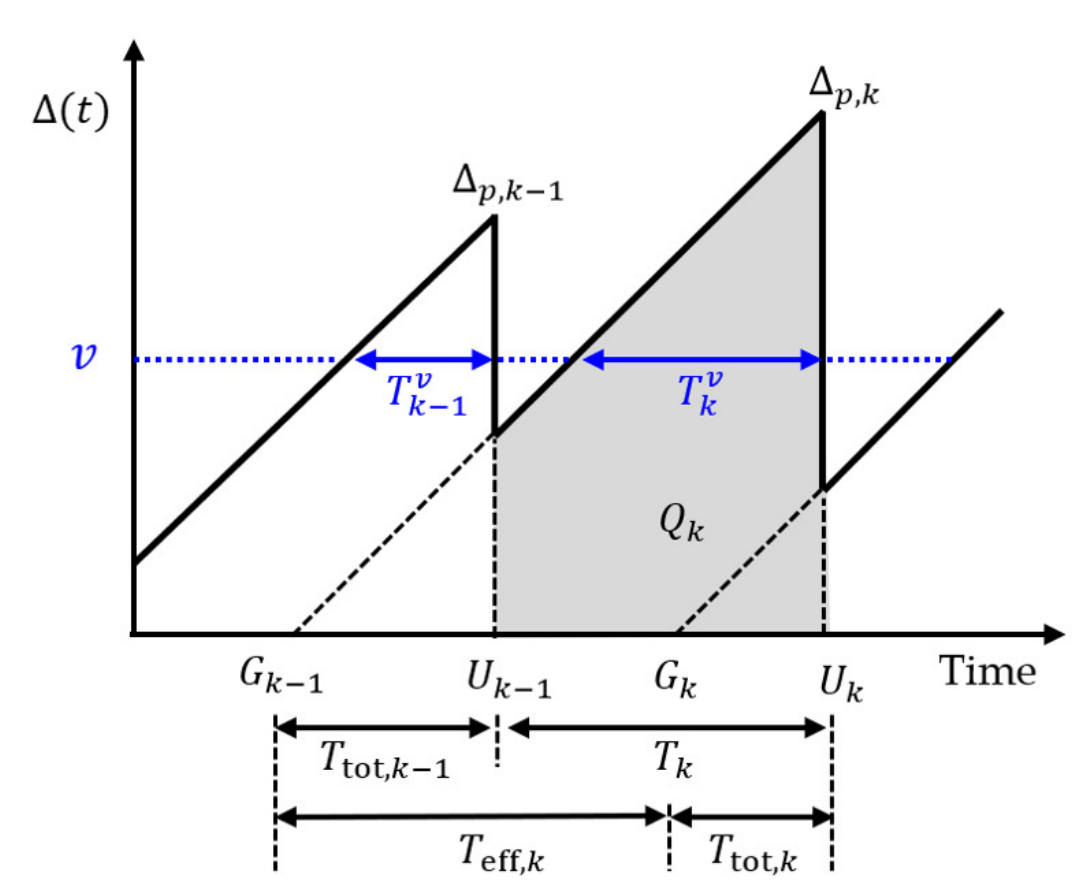}
	\end{center}

	\caption{A sample path of the \ac{AoI} where $G_k$ and $U_k$ are the generation instant and the update instant of effective packet $k$.  }
	
	\label{fig:sample path}
	
\end{figure}

The inter-generation time of two consecutive packets that successfully arrive at the \ac{BS} from the source is denoted by $T_{\text{int}}.$ Since the inter-generation time of packets follows an exponential distribution with rate $\rho_\text{s}$, $\Tint$ will also follow an exponential distribution with rate $\rho = \rho_\text{s} p_c$. We assume independence between $\Tint$ and $\{X_k, \forall k \geq 1\}$ in \ac{BeMN}.\footnote{A dependence exists between $\Tint$ and $\{X_k, \forall k \geq 1\}$ due to the ordering phase in \ac{BeMN}. However, this dependence is negligible as shown in Section \ref{Sec: Numerical Results}-A.} Due to the memoryless property of $T_{\text{int}}$, $\Teff$ can be given as
\begin{align}
\Teff = X_{k-1} + T_{\text{int}}. \label{T_eff}
\end{align}

 The \ac{PAoI} is the \ac{AoI} just before the update instant. As shown in Fig. \ref{fig:sample path}, for effective packet $k$, \ac{PAoI} is   given by
\begin{equation}
\PAoI = \Teff + \TotLat. \label{PAoI_definition}
\end{equation}
From \eqref{T_eff} and \eqref{PAoI_definition}, we derive the average \ac{AoI} \avgaoi{} of \ac{BeMN} in the following lemma.

\begin{lemma}
	In \ac{BeMN}, the average \ac{AoI} \avgaoi \ is given by
	\begin{align}
	& \Avgaoi = \frac{\rho \beta}{2 (\alpha \rho + \beta)}
	\left(
	\frac{2}{\rho^2} + \frac{2\alpha}{\rho \beta} + \frac{\alpha^2 + \alpha}{\beta^2}
	\right) 
	+
	\frac{\alpha}{\beta}
	+
	\TransLat,
	\end{align}
	where $\alpha$ and $\beta$ are defined in \eqref{average_latency}, $\rho = \rho_s p_c$, and $\TransLat$ is given by \eqref{Transmission_latency_closed}.
\end{lemma}  
\begin{proof} 
	Let $Q_k$ be the area of the trapezoid between $U_k$ and $U_{k-1}$ as in Fig. \ref{fig:sample path}. From \cite{YsEbRdCkNs:17}, \avgaoi{} can be given by
	\begin{align}
	\Avgaoi & =
	\limit{T}  \frac{1}{T} \integral{0}{T} \Delta(t) \mathrm{d}{t} =
	\limit{T} \frac{1}{T} \summation{k=1}{N(T)} Q_k \nonumber \\
	& =
	\limit{T} \frac{N(T)}{T} \frac{1}{N(T)}  \summation{k=1}{N(T)} Q_k \nonumber \\ 
	&\overset{(a)}{=}
	\limit{T} \frac{\summation{k=1}{N(T)} Q_k / N(T)}{\summation{k=1}{N(T)} (U_k - U_{k-1}) /N(T) } \nonumber \\
	&\overset{(b)}{=}
	\frac{\E[Q_k] } {\E[T_k]}, \label{average_aoi_fraction}
	\end{align}
	where $N(T)$ is the number of updates until time $T$, defined as  $N(T) = \text{max} \{ k \ | \ U_k \leq T  \}.$ In \eqref{average_aoi_fraction}, \emph{(a)} follows from the fact that $T$ can be presented in an infinite summation of the update intervals $U_k - U_{k-1}$ as $T$ goes to infinity, and \emph{(b)} follows from the ergodicity of the sample path.

	In \eqref{average_aoi_fraction}, we have:
	$Q_k = \frac{1}{2} (\ac{PAoI}_k^2 - \TotLat[k-1]^2)$. Therefore, using \eqref{PAoI_definition}, $\E[Q_k]$ can be obtained as 
	\begin{align}
	\E[Q_k] &=\frac{1}{2} \E[\Teff^2 + 2 \Teff \TotLat].
	\end{align}
	Using $\Teff$ in \eqref{T_eff}, $\E[Q_k]$ can be represented as
	\begin{align}
	\E[Q_k] &= \frac{1}{2} \E[(\Tint \hspace{-0.5mm} + \hspace{-0.5mm} X_{k-1} )^2] \hspace{-0.5mm} + \hspace{-0.5mm} \E[(\Tint \hspace{-0.5mm} + \hspace{-0.5mm} X_{k-1}) (X_k \hspace{-0.5mm} + \hspace{-0.5mm} \TransLat)]. \label{E[Q_k] before}
	\end{align}
	Since $\Tint \sim \text{Exp}(\rho)$ and $X_k \sim \text{Gamma}(\alpha, \beta)$, \eqref{E[Q_k] before} can be given by
	\begin{align}
	\E[Q_k]	&= 
	\frac{2}{\rho^2} + \frac{2 \alpha}{\rho \beta} + \frac{\alpha^2 + \alpha}{\beta^2} + \left(
	\frac{1}{\rho} + \frac{\alpha}{\beta}
	\right)
	\left(
	\frac{\alpha}{\beta} + \TransLat
	\right). \label{E[Q_k]}
	\end{align}
	%
	%
	
	In \eqref{average_aoi_fraction}, as shown in Fig. \ref{fig:sample path}, $T_k$ represents the interval of the update instants $U_{k-1}$ and $U_k$, which can be given by 
	\begin{equation}
	T_k = \Teff + \TotLat - \TotLat[k-1] =X_k + \Tint. \label{E[T_k]}
	\end{equation}
	Hence, we have $\displaystyle \E[T_k] = \frac{1}{\rho} + \frac{\alpha}{\beta}.$ The average \ac{AoI} \avgaoi{} can be then presented by substituting  \eqref{E[Q_k]} and \eqref{E[T_k]} into \eqref{average_aoi_fraction}.
\end{proof}

From Lemma 1, we can obtain \avgaoi{} in \ac{BeMN}, which captures the overall freshness of status information stored in ledgers. 
To see whether a certain level of freshness is guaranteed  in \ac{BeMN}, we can evaluate the \emph{\ac{AoI} violation probability}, which is the probability that the \ac{AoI} exceeds a target \ac{AoI} $v$. The AoI violation probability $\Aoivio$ is given by \cite{JaHuJa:19}
\begin{align}
\Aoivio =  \P[\Delta(t) \geq \TargetAoI] = \frac{\E[T^\TargetAoI_k]}  { \E[T_{k}] }, \label{ratio_of_E}
\end{align}
where $T^\TargetAoI_k$ is the time duration during which the \ac{AoI} is larger than $\TargetAoI$ between the update instants $U_{k-1}$ and $U_k$. Then, $T^\TargetAoI_k$ is given by 
\begin{align}
T^\TargetAoI_k &= \min{ \left\{ \left(\PAoI  -   \TargetAoI \right)^+,   T_k  \right\} }  \nonumber   \\
&= \min{ \left\{ \left( X_{k-1} \hspace{-0.5mm} + \hspace{-0.5mm} X_{k} \hspace{-0.5mm} +  \hspace{-0.5mm}\Tint  \hspace{-0.5mm} + \hspace{-0.5mm} \TransLat  \hspace{-0.5mm}-  \hspace{-0.5mm} \TargetAoI \right)^+,   X_{k}  \hspace{-0.5mm}+ \hspace{-0.5mm} \Tint         \right\}          }, \label{min_T^v_k}
\end{align}
where $(\cdot)^+ = \max {(0, \ \cdot) }$. We now obtain \aoivio{} of \ac{BeMN} in the following theorem.

\begin{figure*}[!t]
	\vspace{-0.2cm}
	\begin{align}
	\Aoivio &= \frac{\rho \beta^{2\alpha+1}}{\left( \beta \hspace{-0.5mm} + \hspace{-0.5mm} \rho \alpha \right) \Gamma(\alpha)^2} \summation{n=0}{\infty} \frac{\left( \rho \hspace{-0.5mm}  - \hspace{-0.5mm} \beta  \right)^n}{n! (\alpha \hspace{-0.5mm} + \hspace{-0.5mm} n )   \rho^{\alpha +n+1} }
	\Bigg[
	\frac{\Gamma(\alpha \hspace{-0.5mm} + \hspace{-0.5mm} n \hspace{-0.5mm} + \hspace{-0.5mm} 1)}{\beta^\alpha} \gamma\left(\alpha, \beta \Remain   \right) - \summation{k=0}{\infty} \frac{(-1)^k (\rho \Remain)^{\alpha +  n   +  k   +  1}}{k! \left( \alpha  \hspace{-0.5mm} + \hspace{-0.5mm} n \hspace{-0.5mm} + \hspace{-0.5mm}  k \hspace{-0.5mm} +  \hspace{-0.5mm} 1 \right) } B\left( \alpha  \hspace{-0.5mm} + \hspace{-0.5mm} n \hspace{-0.5mm} +  \hspace{-0.5mm}  k \hspace{-0.5mm}  + \hspace{-0.5mm} 2 , \alpha  \right)
	\nonumber \\ 
	& \quad
	\times  \Remain^{\alpha}   {_1}F_1 \left(\alpha; 2\alpha  \hspace{-0.5mm} + \hspace{-0.5mm} n \hspace{-0.5mm} + \hspace{-0.5mm} k  \hspace{-0.5mm} + \hspace{-0.5mm}  2; -\beta \Remain \right)
	\Bigg]
	\hspace{-0.5mm} + \hspace{-0.5mm}  \frac{\rho}{(\beta \hspace{-0.5mm}  + \hspace{-0.5mm} \rho \alpha) \Gamma(\alpha)}  \left\{ \alpha \gamma(\alpha, \beta \Remain)  \hspace{-0.5mm} -  \hspace{-0.5mm}   \summation{n=0}{\infty} \frac{ (-\beta \Remain)^{2\alpha+n+1}} {n! \Gamma(\alpha) (\alpha  \hspace{-0.5mm} + \hspace{-0.5mm} n \hspace{-0.5mm}  + \hspace{-0.5mm} 1)} B(\alpha \hspace{-0.5mm} + \hspace{-0.5mm} n \hspace{-0.5mm} + \hspace{-0.5mm} 2, \alpha) \right. \nonumber \\
	& \quad \left. \times {_1}F_1 (\alpha; 2\alpha \hspace{-0.5mm} + \hspace{-0.5mm} n \hspace{-0.5mm} + \hspace{-0.5mm} 2; -\beta \Remain) - (\beta \Remain)^{\alpha+1}  B(\alpha, 2) \Remain^{\alpha} {_1}F_1 (\alpha; \alpha  \hspace{-0.5mm} + \hspace{-0.5mm} 2; -\beta \Remain) \right. \nonumber \\
	&\quad \left. +  \summation{n=0}{\infty} \frac{ (-\beta \Remain)^{2\alpha+n+1}}{  n! \Gamma(\alpha) (\alpha+n)} B(\alpha , \alpha \hspace{-0.5mm}  +\hspace{-0.5mm} n \hspace{-0.5mm} + \hspace{-0.5mm} 2) {_1}F_1(\alpha; 2\alpha \hspace{-0.5mm} + \hspace{-0.5mm} n \hspace{-0.5mm} + \hspace{-0.5mm} 2; \beta \Remain) \right\} + \frac{\Gamma(\alpha, \beta \Remain)}{\Gamma(\alpha)}. \label{AoI_violation_Prob}
	\end{align}
	\centering \rule[0pt]{18cm}{0.3pt}
\end{figure*}

\begin{theorem}
	In \ac{BeMN}, the \ac{AoI} violation probability \aoivio{} is given by \eqref{AoI_violation_Prob} (see the top of the next page), where $\Remain =(\TargetAoI-\TransLat)^+$.		\label{theorm 1}	
\end{theorem}
\begin{proof}
	See Appendix \ref{app:theorem 1}.
\end{proof}

From Theorem~1, we can identify how the tail of the \ac{AoI} distribution in \ac{BeMN} is formed, which is not captured by \avgaoi. This result can be used to prevent the status update \ac{AoI} from being too stale. In Section \ref{Sec: Numerical Results}-A, we further show that the optimal value of the \ac{BeMN} parameters to minimize \aoivio\ can be different from \avgaoi. 

We also obtain the \ac{AoI} violation probability for the case of integer $\alpha$ in the following corollary.

\begin{cor}	
When the shape parameter $\alpha$ of the consensus latency distribution is an integer, 
the \ac{AoI} violation probability \aoivio\ is given by
	\begin{align}
	&\Aoivio =
\frac{\beta^{2 \alpha + 1} e^{-\rho \Remain} \gamma
	\left( \hspace{-0.5mm}
	\alpha, \Remain \left( \beta - \rho \right) \hspace{-0.5mm}
	\right) \hspace{-0.5mm}  }
{\left(\alpha \rho \hspace{-0.5mm}  + \hspace{-0.5mm}  \beta   \right) \left( \beta - \rho  \right)^{2 \alpha} \Gamma(\alpha) } 
+ \frac{ \Gamma(\alpha, \beta \Remain)}{\Gamma(\alpha)}
\nonumber \\
& \quad \quad +  \hspace{-0.5mm}
\frac{\rho}{\alpha \rho \hspace{-0.5mm} + \hspace{-0.5mm} \beta} \hspace{-1.2mm}
\summation{m=0}{\alpha-1}  \hspace{-0.5mm}
\summation{k=0}{m} \hspace{-0.5mm}
\frac{(\beta \Remain)^{\alpha + k} }{\Gamma(\alpha \hspace{-0.7mm} + \hspace{-0.7mm} k + \hspace{-0.7mm}  1)}  
e^{-\beta \Remain} \hspace{-0.5mm}
\left\{ \hspace{-0.5mm}
1 \hspace{-0.5mm} - \hspace{-0.5mm} (\rho \hspace{-0.5mm} - \hspace{-0.5mm} \beta )^{m \hspace{-0.2mm} - \hspace{-0.2mm} \alpha}
\hspace{-0.2mm} \right\} \hspace{-0.5mm} . \label{Intger_AoI_vio}
	\end{align}
	\label{cor 1}
\end{cor}
\begin{proof}
		Note that $\gamma(\alpha, x)$ and $\Gamma(\alpha, x)$ can be represented as finite series when the shape parameter $\alpha$ is a positive integer value as shown in \cite[Equation. 3.351-2]{ToI}. In other words, we have:
	\begin{align}
	\gamma(\alpha, x) \hspace{-0.3mm} = \hspace{-0.3mm} \Gamma(\alpha) \hspace{-0.7mm}
	\left( \hspace{-0.7mm} 
	1 \hspace{-0.5mm} - \hspace{-0.5mm} \summation{n=0}{\alpha - 1} \frac{x^n e^{-x} }{n!} \hspace{-0.7mm}
	\right) \hspace{-0.7mm}
	, \Gamma(\alpha, x) = \Gamma(\alpha) \hspace{-0.7mm}
	\summation{n=0}{\alpha - 1} \hspace{-0.5mm}
	\frac{x^n e^{-x} }{n!}. 	\label{3.351-1, 2}
	\end{align}
	
	Using \eqref{3.351-1, 2}, $\P[\Tint \geq a + \Remain - x - X_k]$ in \eqref{Just Prob} can be obtained as follows:
	\begin{align}
	&\P \left[ \Tint  \geq a + \Remain - x - X_k \right] \nonumber \\
	&= \frac{\beta^{\alpha} 	e^{-\rho \left(a + \Remain \hspace{-0.5mm} - x  \right)}  }  {\left( \beta  	\hspace{-0.5mm} -  	\hspace{-0.5mm} \rho  \right)^{\alpha}}
	\hspace{-0.5mm}
	\left\{
	\hspace{-0.5mm}
	1  
	\hspace{-0.9mm}	
	- 
	\hspace{-0.9mm}	
	\summation{m=0} {\alpha-1} \hspace{-0.7mm} \frac{(\beta 	\hspace{-0.5mm} - 	\hspace{-0.5mm} \rho)^m}{m!}
	\hspace{-0.5mm}	
	\left(
	a 	\hspace{-0.5mm} + 	\hspace{-0.5mm} \Remain 	\hspace{-0.8mm} - 	\hspace{-0.5mm} x  
	\right)^m  \hspace{-0.5mm}
	e^{ -\frac{\beta}{\rho} } \hspace{-0.5mm}
	\right\}
	\nonumber  \\
	& \quad	+
	\summation{m=0}{\alpha-1} \frac{\beta^m}{m!} \left(  a + \Remain - x  \right)^m e^{ -\beta \left( a + \Remain - x  \right) }.  \label{Integer Just Prob}
	\end{align}
	When computing $\E[T_k^\TargetAoI]$ in \eqref{T^v_k}, for obtaining $E_1(x)$, the integral of \eqref{Integer Just Prob}, we can first derive the following: 
	\begin{align}
	& \integral{0}{\infty} 
	\P \left[
	\Tint  \geq a + \Remain - x - X_k
	\right] \mathrm{d}a \nonumber \\
	& = 
	\integral{0}{\infty}  \hspace{-0.5mm}
	\left[ 
	\frac{\beta^\alpha e^{-\rho \left(a + \Remain - x  \right)}}  {\left( \beta - \rho  \right)^\alpha} \hspace{-0.5mm}
	- \hspace{-0.5mm}
	\left\{
	\summation{m=0}{\alpha-1}
	\left(
	\frac{\beta}{(\beta - \rho)}
	\right)^\alpha
	\frac{(\beta - \rho)^m}{m!} 
	\right. \right.
	\nonumber \\
	&\left. \left.
	\quad \quad \quad \quad - \frac{\beta^m}{m!}
	\right\}
	\times
	\left(  a + \Remain - x  \right)^m e^{ -\beta \left( a + \Remain - x  \right) }
	\right] \mathrm{d}a \nonumber \\
	&=
	\frac{\beta^\alpha e^{-\rho \left(\Remain - x  \right)}}  {\rho \left( \beta - \rho  \right)^\alpha} \hspace{-0.5mm}
	-              \hspace{-0.5mm} 
	\summation{m=0}{\alpha-1} \hspace{-0.5mm}
	\frac{\Gamma  
		\left(
		m+1, \beta \left( \Remain - x \right)
		\right)
	} {\beta \hspace{0.3mm} m!}  \nonumber \\
	&\quad \quad \quad \quad \quad \quad \quad \quad \quad \times 
	\left\{
	1 - (1 - \rho / \beta)^{m - \alpha}  
	\right\},  \label{Integer substitution}
	\end{align}
	where the last equation is obtained by substituting $k$ for $a + \Remain - x$ and \eqref{definition_of_incomplete}. Then, using \eqref{Integer substitution}, $E_1(x)$ can be obtained by 
	\begin{align}
	&\integral{0}{\Remain} \hspace{-2.0mm}
	\integral{0}{\infty} \hspace{-1.5mm} 
	\P
	\left[
	\Tint  \geq a + \Remain - x - X_k 
	\right] f_{X_{k-1}}(x) \ \mathrm{d}a \mathrm{d}x \nonumber \\
	& \overset{(a)}{=}
	\frac{e^{-\rho \Remain}  \beta ^ {2 \alpha } \gamma 
		\left( 
		\alpha, \Remain \left( \beta - \rho \right)
		\right)} 
	{\rho \hspace{0.3mm} \Gamma(\alpha) \left( \beta - \rho  \right)^{2 \alpha}} \hspace{-0.5mm}
	+ \hspace{-0.5mm}
	\summation{m=0}{\alpha-1} \hspace{-0.5mm}
	\summation{k=0}{m} \hspace{-0.5mm}
	\frac{\beta^{\alpha +k - 1} \Remain^{\alpha +k}}{\Gamma(\alpha +k +1)}
	\nonumber \\
	& \quad \quad \times
	e^{-\beta \Remain} 
	\left\{
	1 - (1 - \rho/\beta)^{m - \alpha}
	\right\}, \label{Integer T^v_k first}
	\end{align}
	where $(a)$ follows from the representation of $\gamma(\alpha, x)$ in \eqref{3.351-1, 2} and the fact that a beta function $B(a,b)$ is the same as $\Gamma(a) \Gamma(b) / \Gamma(a+b)$ \cite[Equation 8.384-1]{ToI}. 
	From \eqref{ratio_of_E}, \eqref{T^v_k}, \eqref{T_ak_2}, and \eqref{Integer T^v_k first}, \aoivio{} with integer $\alpha$ can be obtained as \eqref{Intger_AoI_vio}.
\end{proof}

\begin{remark}
From Corollary~\ref{cor 1}, we can see that \aoivio\ becomes one as the target \ac{STP} $\TargetSTP$\ approaches one (i.e.,  $\TargetSTP = 1$) since $\Remain =(\TargetAoI-\TransLat)^+$ approaches zero. This means that, in order to satisfy a higher target \ac{STP}, the \ac{AoI} violation will always happen because the transmission latency $\TransLat$ becomes large. Hence, we can see that a higher target \ac{STP} does not always guarantee a lower \ac{AoI} violation probability.

For the case of non-integer $\alpha$, 
\aoivio\ in \eqref{Intger_AoI_vio} 
can serve as an upper or lower bound of the AoI violation probability
by using $\hat{\alpha}=$ Ceil($\alpha$) or Floor($\alpha$) instead of $\alpha$, respectively, where Ceil($\cdot$) and Floor($\cdot$) are the ceil and floor functions.
This is because the \ac{CCDF} of the consensus latency can be upper bounded by the one with a higher shape parameter under the same rate parameter, and vice versa.  
\end{remark}

\begin{figure}
	\begin{center}   
		{ 
			\psfrag{A1111111111}[Bl][Bl][0.59]   {$\hat{\alpha}$ = Floor ($\alpha$)}
			\psfrag{A2}[Bl][Bl][0.59]   {$\hat{\alpha}$ = Ceil ($\alpha$)}
			\psfrag{A3}[Bl][Bl][0.59]   {$\alpha$ (Exact)}
			\psfrag{Y111111111}[bc][bc][0.9] {\ac{AoI} violation probability, $\Aoivio$}
			\psfrag{X111111111}[tc][tc][0.9] {Target AoI, $\TargetAoI$}
			\includegraphics[width=1.00\columnwidth]{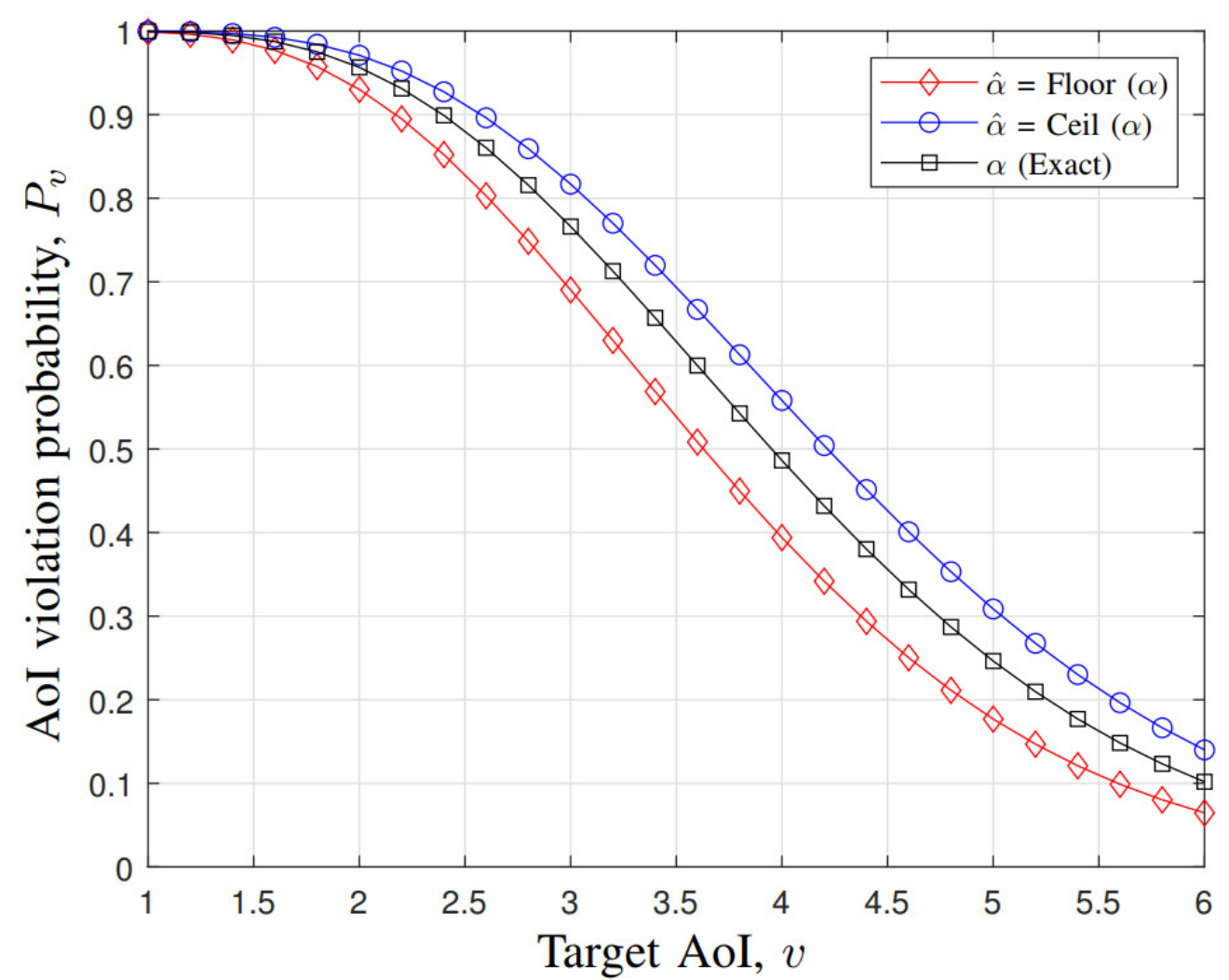}
		}
	\end{center}
	\caption{
		\ac{AoI} violation probability \aoivio{} as a function of a target \ac{AoI} $\TargetAoI$ for different shape parameters.}
	
	\label{fig: bound }
\end{figure}

In Fig. \ref{fig: bound }, we compare the result of Corollary 1 as a function of $\TargetAoI$ for target \ac{STP} $\TargetSTP = 0.5$. The shape parameter ($\alpha = 5.55$) is mapped to integers by the ceil and floor function. As shown in Fig. \ref{fig: bound }, we first observe that the three results show a similar trend. We can also see that the result of Corollary 1 can be the upper or lower bound of Theorem 1 as discussed in Remark 1. 

We now derive the \ac{PAoI} violation probability, which is the probability that the \ac{PAoI} is larger than or equal to a target value $\TargetAoI$, i.e., $\Paoivio = \P[\PAoI \geq \TargetAoI]$. This metric is particularly useful for a system that must maintain the worst-case \ac{AoI} below a certain value as much as possible. In the following lemma, we analyze $\Paoivio$. 

\begin{lemma}
	In \ac{BeMN}, the \ac{PAoI} violation probability \paoivio{} is given by
	\begin{align}
	&\Paoivio = 1 - \frac{\beta^{2 \alpha}}{\Gamma(\alpha)^2} \hspace{-0.5mm} \summation{n=0}{\infty} \hspace{-0.5mm} \summation{k=0}{\infty} \hspace{-0.5mm}
	\frac{\Remain^{2\alpha +n +k} }{n! k! (\alpha + n)} B( \alpha \hspace{-0.5mm} + \hspace{-0.2mm} n \hspace{-0.5mm} + \hspace{-0.5mm} 1, \alpha \hspace{-0.5mm}  + \hspace{-0.5mm}  k \hspace{-0.2mm} ) 
	\nonumber \\
	&\quad \quad \quad \quad \quad \quad \quad \quad \quad \quad
	\times \hspace{-0.5mm}
	\left\{ \hspace{-0.5mm}
	(-\beta)^{n+k} \hspace{-0.5mm}
	- \hspace{-0.5mm} 
	e^{-\rho  \Remain} (\rho \hspace{-0.5mm} - \hspace{-0.5mm} \beta) ^{n+k}   
	\right\} \hspace{-0.5mm} ,
    \label{PAoI_violation}
	\end{align}
	where $\Remain = (\TargetAoI - \TransLat)^+$.
	
\end{lemma}

\begin{proof}
	%
 	From \eqref{PAoI_definition}, the \ac{CDF} of $\PAoI$ is defined as
	\begin{align}
	F_{\small  \PAoI }(\TargetAoI) 
	&=  \P[\Teff \hspace{-0.5mm} + \hspace{-0.5mm} \TotLat \leq \TargetAoI] \nonumber \\
	& \overset{(a)}{=} \P[X_{k-1} +X_k + \Tint + \TransLat \leq \TargetAoI] , 
	\label{PAoI_CDF}
	\end{align}	
	where $(a)$ is from \eqref{total_latency} and \eqref{T_eff}. We then have
	\begin{align}
	& F_{\small  \PAoI }(\TargetAoI) = \hspace{-1.5mm} \integral{0}{\Remain} \hspace{-1.7mm} \P[x + X_k + \Tint \leq \TargetAoI - \TransLat ]  f_{X_{k-1}}(x)\mathrm{d}x  \nonumber \\
	&=  \hspace{-1.6mm}
	\integral{0}{\Remain}  \hspace{-2.5mm}
	\integral{0}{\Remain \hspace{-0.8mm} -   x} \hspace{-1.7mm} 
	\left\{
	\hspace{-0.5mm} 1 \hspace{-0.5mm} - \hspace{-0.5mm} e^{-\rho \left( \Remain \hspace{-0.8mm} - \hspace{-0.2mm} x \hspace{-0.2mm} -\hspace{-0.2mm} x'   \right) } \hspace{-0.5mm}
	\right\}
	\hspace{-0.9mm} f_{X_k} \hspace{-0.5mm} (x' \hspace{-0.2mm}) f_{X_{k-1}} \hspace{-0.5mm} (x) \mathrm{d}x' \mathrm{d}x \nonumber \\
	&= \frac{1}{\Gamma(\alpha)} \integral{0}{\Remain} 
	\bigg[
	\gamma\left(
	\alpha, \beta \left( \Remain - x  \right)    
	\right) - \frac{\beta^\alpha e^{-\rho \left( \Remain - x \right)}  } { \left( \beta - \rho \right)^\alpha }   \nonumber \\
	& \quad \quad \quad \quad \quad \quad \ \times \hspace{-0.7mm}  \gamma \hspace{-0.7mm}
	\left( 
	\alpha, \hspace{-0.5mm} \left( \beta \hspace{-0.5mm} - \hspace{-0.5mm} \rho   \right) \hspace{-0.5mm}  \left( \Remain \hspace{-0.5mm} - \hspace{-0.5mm} x  \right)   
	\right) \hspace{-0.5mm}  \bigg] \hspace{-0.5mm}  f_{X_{k-1}} \hspace{-0.5mm} (x) \mathrm{d}x.
	\label{PAoI_1}
	\end{align}
	Using the Taylor series of $\displaystyle \exp(-\beta(\Remain - x))$ and \eqref{8.354-1}, $F_{\small  \ac{PAoI}_k }(\TargetAoI)$ in \eqref{PAoI_1} can be given by
	\begin{align}
	&F_{\small  \PAoI }(\TargetAoI) \nonumber \\
	&=\frac{\beta^ {2 \alpha}}  {\Gamma(\alpha)^2 } \hspace{-1.0mm} \summation{n=0}{\infty} \hspace{-0.5mm} \summation{k = 0} { \infty }  \hspace{-0.5mm} 	
	\bigg[
	\frac{ 1} { n! k! \left( \alpha +n  \right) } \hspace{-0.5mm}  \integral{0}{\Remain} \hspace{-3.5mm}  \left(  \Remain - x \right)^n x^{\alpha + k - 1} 
	\nonumber \\
	& 
	\quad \hspace{-0.5mm}
	\times 
	\left\{
	(-\beta)^{n+k} - 
	e^{-\rho \Remain   }   \left( \rho \hspace{-0.5mm} - \hspace{-0.5mm} \beta \right)^{n+k} 
	\hspace{-0.7mm} (\Remain  \hspace{-0.5mm} - \hspace{-0.5mm} x)^{\alpha} \mathrm{d}x 
	\right\}
	\bigg]. \label{PAoI_3}
	\end{align}
	Using ${\int_{0}^{\alpha} x^{\beta-1} (\alpha - x)^{m-1} \mathrm{d}x = \alpha^{m+\beta-1} B(m, \beta) }$ \cite[Equation 3.191-1]{ToI} and \eqref{gamma_pdf}, $F_{\small  \ac{PAoI}_k }(\TargetAoI)$ in \eqref{PAoI_3} is given by
	\begin{align}
	&F_{\small  \PAoI }(\TargetAoI) = \frac{\beta^{2 \alpha}}{\Gamma(\alpha)^2} \hspace{-0.5mm} \summation{n=0}{\infty} \hspace{-0.5mm} \summation{k=0}{\infty} \hspace{-0.5mm}
	\frac{\Remain^{2\alpha +n +k} }{n! k! (\alpha + n)} B( \alpha \hspace{-0.5mm} + \hspace{-0.2mm} n \hspace{-0.5mm} + \hspace{-0.5mm} 1, \alpha \hspace{-0.5mm}  + \hspace{-0.5mm}  k \hspace{-0.2mm} ) 
	\nonumber \\
	&\quad \quad \quad \quad \quad 
	\times
	\left\{
	(-\beta)^{n+k}
	- e^{-\rho \Remain  } (\rho - \beta) ^{n+k}   
	\right\}.
	\label{before derivation}
	\end{align}
	By taking the complement rule to \eqref{before derivation}, we obtain \paoivio{} as \eqref{PAoI_violation}.
\end{proof}

Lemma~2 can show the distribution of the worst-case of data freshness during update intervals. This result can be used to control the tail of the worst-case \ac{AoI}, especially with applications that require a strict threshold of data freshness.

	\section{Experimental and Numerical Results} 
 	\label{Sec: Numerical Results}
 	In this section, we present experimental and numerical results to verify the analysis of the \ac{AoI} violation probability \aoivio\ and to show the impact of \ac{BeMN} parameters, i.e., target \ac{STP} $\TargetSTP$, block size \blocksize, and timeout \timeout, on data freshness. 

 	 	\begin{table*}[t!]
 		\caption{Average estimated shape and rate parameters $(\alpha, \beta)$ of the  consensus latency distribution and the average KS statistics for different target \ac{STP} $\TargetSTP$ and block size \blocksize.} 
 		
 		\renewcommand{\arraystretch}{1.3}
 		\resizebox{1 \textwidth}{!}{%
 			\begin{tabular}{ccccclccccc}
 				\cline{1-5} \cline{7-11}
 				$\TargetSTP$ & \begin{tabular}[c]{@{}c@{}}Average estimate\\ $(\alpha, \beta)$\end{tabular} & \begin{tabular}[c]{@{}c@{}}Average\\ latency/SD\end{tabular} & \begin{tabular}[c]{@{}c@{}}Average\\ skewness\end{tabular} & \begin{tabular}[c]{@{}c@{}}Average\\ KS statistics\end{tabular} &  & \blocksize & \begin{tabular}[c]{@{}c@{}}Average estimate\\ $(\alpha, \beta)$\end{tabular} & \begin{tabular}[c]{@{}c@{}}Average\\ latency/SD\end{tabular} & \begin{tabular}[c]{@{}c@{}}Average\\ skewness\end{tabular} & \begin{tabular}[c]{@{}c@{}}Average\\ KS statistics\end{tabular} \\ \cline{1-5} \cline{7-11} 
 				0.3     & (5.64, 3.01)                                                                 & 2.42/0.95                                                    & 0.093                                                      & 0.0732                                                          &  & 3         & (1.62, 0.30)                                                                 & 5.71/4.03                                                    & 0.342                                                      & 0.0831                                                          \\
 				0.4     & (5.94, 2.45)                                                                 & 2.42/0.92                                                    & 0.086                                                      & 0.0623                                                          &  & 5         & (2.90, 1.38)                                                                 & 2.16/1.25                                                    & 0.182                                                      & 0.0333                                                          \\
 				0.5     & (5.39, 2.85)                                                                 & 2.17/0.87                                                    & 0.095                                                      & 0.0506                                                          &  & 7         & (4.35, 2.58)                                                                 & 1.70/0.85                                                    & 0.121                                                      & 0.0498                                                          \\
 				0.6     & (5.42, 2.84)                                                                 & 1.90/0.76                                                    & 0.097                                                      & 0.0504                                                          &  & 10        & (5.24, 3.30)                                                                 & 1.59/0.74                                                    & 0.099                                                      & 0.0495                                                          \\
 				0.7     & (7.18, 3.73)                                                                 & 1.92/0.67                                                    & 0.071                                                      & 0.0462                                                           &  & 12        & (5.81, 3.66)                                                                 & 1.58/0.63                                                    & 0.074                                                      & 0.0382                                                          \\
 				0.8     & (7.71, 4.12)                                                                 & 1.87/0.63                                                    & 0.066                                                      & 0.0423                                                           &  & 15        & (6.95, 3.85)                                                                 & 1.80/0.65                                                    & 0.074                                                      & 0.0381                                                          \\
 				0.9     & (7.50, 4.35)                                                                 & 1.73/0.60                                                    & 0.068                                                      & 0.0369                                                           &  & 20        & (5.42, 2.84)                                                                 & 1.90/0.76                                                    & 0.097                                                      & 0.0504                                                          \\
 				1.0     & (6.57, 3.82)                                                                 & 1.76/0.76                                                    & 0.085                                                      & 0.0532                                                          &  & 25        & (4.85, 2.36)                                                                 & 2.05/0.86                                                    & 0.107                                                      & 0.0604                                                          \\ \cline{1-5} \cline{7-11} 
 			\end{tabular}
 		}

 		\label {Table 1}

 	\end{table*}%

 	\begin{table}[t!]
 		\caption{Average estimated shape and rate parameters $(\alpha, \beta)$ of the  consensus latency distribution and the average KS statistics for different timeout \timeout.} 
 		\begin{center}
 			\renewcommand{\arraystretch}{1.3}

 			\resizebox{0.5 \textwidth}{!}{%
 				\begin{tabular}{ccccc}
 					\hline
 					\timeout & \begin{tabular}[c]{@{}c@{}}Average estimate\\ $(\alpha, \beta)$\end{tabular} & \begin{tabular}[c]{@{}c@{}}Average\\ latency/SD\end{tabular} & \begin{tabular}[c]{@{}c@{}}Average\\ skewness\end{tabular} & \begin{tabular}[c]{@{}c@{}}Average\\ KS statistics\end{tabular} \\ \hline
 					0.5       & (2.74, 0.89)                                                                 & 3.08/2.00                                                    & 0.194                                                      & 0.0420                                                          \\
 					0.6       & (4.26, 2.04)                                                                 & 2.10/1.07                                                    & 0.122                                                      & 0.0452                                                          \\
 					0.7       & (8.28, 5.40)                                                                 & 1.53/0.54                                                    & 0.061                                                      & 0.0494                                                          \\
 					0.75      & (6.78, 5.19)                                                                 & 1.30/0.47                                                    & 0.076                                                      & 0.0489                                                          \\
 					1.0       & (6.96, 4.65)                                                                 & 1.50/0.54                                                    & 0.075                                                      & 0.0588                                                          \\
 					1.25      & (9.62, 5.37)                                                                 & 1.79/0.55                                                    & 0.053                                                      & 0.0446                                                          \\
 					1.50      & (9.86, 5.20)                                                                 & 1.89/0.56                                                    & 0.052                                                      & 0.0603                                                          \\
 					2.0       & (6.79, 3.62)                                                                 & 1.87/0.66                                                    & 0.075                                                      & 0.0535                                                          \\
 					2.5       & (5.64, 3.01)                                                                 & 1.97/0.72                                                    & 0.091                                                      & 0.0497                                                          \\
 					3.0       & (5.42, 2.84)                                                                 & 1.89/0.76                                                    & 0.097                                                      & 0.0504                                                          \\
 					3.5       & (5.39, 2.85)                                                                 & 1.89/0.75                                                    & 0.091                                                      & 0.0503                                                          \\ \hline
 				\end{tabular}
 			}

 			\label{Table 2}
 		\end{center}

 	\end{table}

 	\subsection{Simulation Setup and KS Test}
 	For our simulations, unless otherwise specified, we use $\rho_s = 15$, $P = 1$ W, $N_0 = -100$ dBm, $W = 1$ MHz, $D = 500$ kb, $\lambda = 0.0001$ (BS/km$^2$), $l = 37$ m, $\TargetSTP=0.6$, $\Blocksize = 20$, and $\Timeout = 3$. We implement an \ac{HLF} platform with version 1.3 \cite{Git} on one physical machine with Intel(R) Xeon W-2155 @ 3.30GHz with 16 GB of RAM.
 	The established \ac{HLF} consists of one peer and two committing peers. Note that the committing peers only verify new blocks conveyed from the ordering service. We generate transactions to update a certain target
 	key-value, which occupies 30 percent of the whole generated transactions. We measure the consensus latency of the generated target transactions with varying \ac{BeMN} parameters. To investigate the effect of each parameter, we fit statistical distribution of thousand transactions for different \ac{BeMN} parameters. Note that we use the maximum likelihood estimation method introduced in Section \ref{Consensus Latency Modeling}. The accuracy of the consensus latency modeling is investigated by using the \ac{KS} test in the following.
  	
 	The \ac{KS} test returns the absolute value of the largest discrepancy between an empirical and theoretical cumulative distribution, known as the \ac{KS} statistic \cite{ANGTANG}. A smaller \ac{KS} statistic means a higher accuracy of the modeling. The \ac{KS} test compares the \ac{KS} statistic with a critical value, which is determined by the number of samples and a significance level. If the \ac{KS} statistic is smaller than the critical value, it is seen as a reasonable theoretical model for the empirical distribution. From the generated transactions above, the corresponding estimated parameters $(\alpha, \beta)$ of the consensus latency and the \ac{KS} statistics with a significance level of 0.01 are averaged over five runs and presented in Tables \ref{Table 1} and \ref{Table 2}. We also present the average \ac{SD} and the skewness of the measured consensus latencies from the experiments. Note that the critical value of the \ac{KS} test for 1000 samples is 0.0515.

 	\begin{figure}
 		\centering
 		\begin{center}   
 			{ 
 				\psfrag{aaaaaaaaaaaaaaaaaaaaa}[Bl][Bl][0.59]   {Experiment, $\TargetSTP$ = 0.4}
 				\psfrag{a2}[Bl][Bl][0.59]   {Simulation, $\TargetSTP$ = 0.4}
 				\psfrag{a3}[Bl][Bl][0.59]   {Analysis, $\TargetSTP$ = 0.4}
 				\psfrag{b1}[Bl][Bl][0.59]   {Experiment, $\TargetSTP$ = 0.6}
 				\psfrag{b2}[Bl][Bl][0.59]   {Simulation, $\TargetSTP$ = 0.6}
 				\psfrag{b3}[Bl][Bl][0.59]   {Analysis, $\TargetSTP$ = 0.6}
 				\psfrag{c1}[Bl][Bl][0.59]   {Experiment, $\TargetSTP$ = 0.8}
 				\psfrag{c2}[Bl][Bl][0.59]   {Simulation, $\TargetSTP$ = 0.8}
 				\psfrag{c3}[Bl][Bl][0.59]   {Analysis, $\TargetSTP$ = 0.8}
 				\psfrag{Y1111111111}[bc][bc][0.9] {\ac{AoI} violation probability, \aoivio}
 				\psfrag{X1111111111}[tc][tc][0.9] {Target AoI, $\TargetAoI$}
 				\includegraphics[width=1.00\columnwidth]{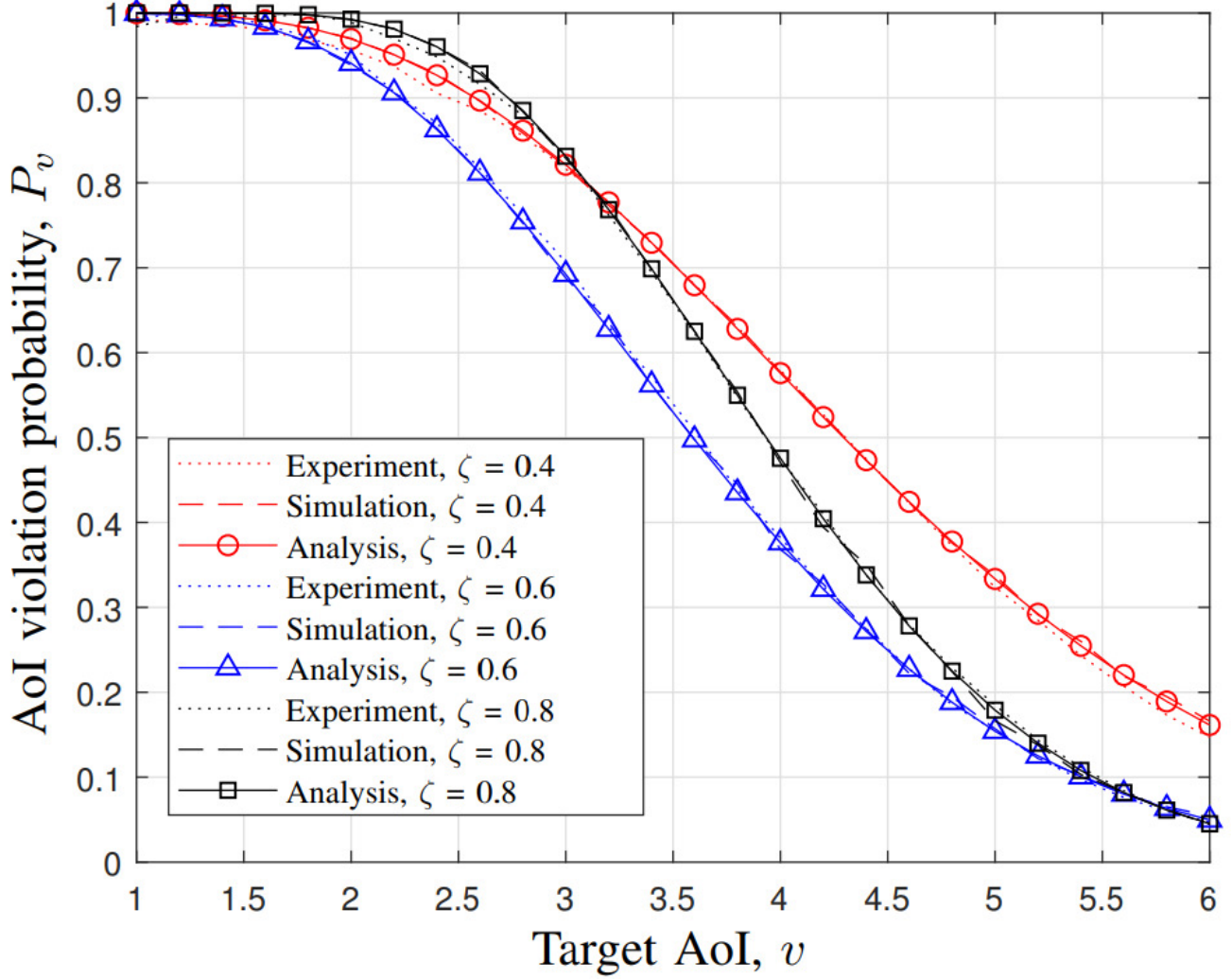}
 			}
 		\end{center}
 		\caption{\ac{AoI} violation probability \aoivio{} as a function of a target \ac{AoI} $\TargetAoI$ for different target \ac{STP}s $\TargetSTP$.}
 		
 		\label{fig: STP_validation}	
 	\end{figure}

 	\begin{figure}
 		\centering
 		\begin{center}
 			\psfrag{aaaaaaaaaaaaaaaaaaaaa}[Bl][Bl][0.59]   {Experiment, \blocksize = 5 }
 			\psfrag{a2}[Bl][Bl][0.59]   {Simulation, \blocksize = 5}
 			\psfrag{a3}[Bl][Bl][0.59]   {Analysis, \blocksize = 5}
 			\psfrag{b1}[Bl][Bl][0.59]   {Experiment, \blocksize = 12 }
 			\psfrag{b2}[Bl][Bl][0.59]   {Simulation, \blocksize = 12}
 			\psfrag{b3}[Bl][Bl][0.59]   {Analysis, \blocksize = 12}
 			\psfrag{c1}[Bl][Bl][0.59]   {Experiment, \blocksize = 25}
 			\psfrag{c2}[Bl][Bl][0.59]   {Simulation, \blocksize = 25}
 			\psfrag{c3}[Bl][Bl][0.59]   {Analysis, \blocksize = 25}
 			\psfrag{Y1111111111}[bc][bc][0.9] {\ac{AoI} violation probability, \aoivio}
 			\psfrag{X1111111111}[tc][tc][0.9] {Target AoI, $\TargetAoI$}
 			\includegraphics[width=1.0\columnwidth]{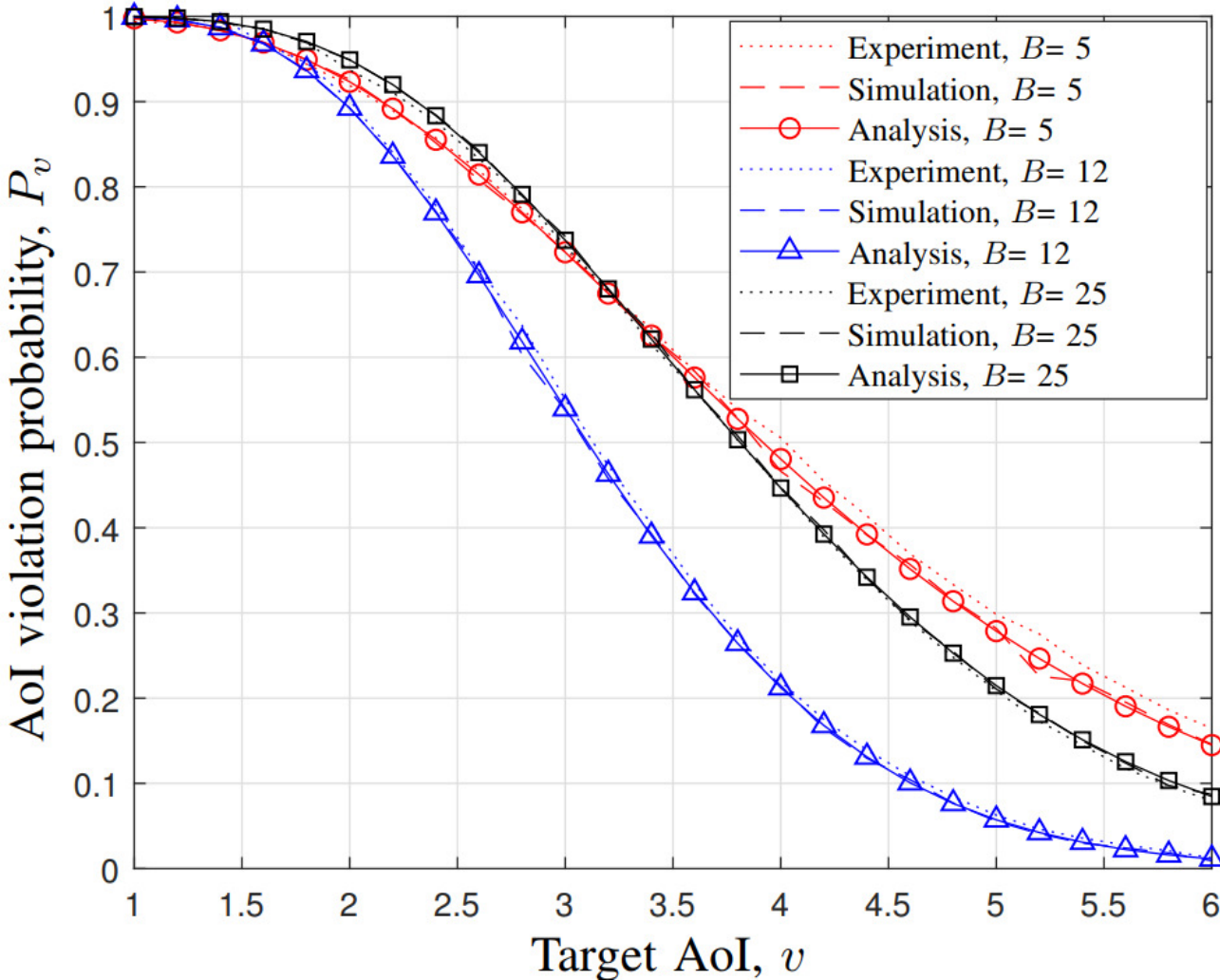}
 		\end{center}
 		
 		\caption{\ac{AoI} violation probability \aoivio{} as a function of target \ac{AoI} $\TargetAoI$ for different block size \blocksize.}
 		
 		\label{fig: block_validation}	
 		\vspace{-0.2cm}
 	\end{figure}

 	\begin{figure}
 		\centering
 		\begin{center}
 			\psfrag{aaaaaaaaaaaaaaaaaa}[Bl][Bl][0.59]   {Experiment, \timeout = 0.5}
 			\psfrag{a2}[Bl][Bl][0.59]   {Simulation, \timeout = 0.5}
 			\psfrag{a3}[Bl][Bl][0.59]   {Analysis, \timeout = 0.5}
 			\psfrag{b1}[Bl][Bl][0.59]   {Experiment, \timeout = 0.75}
 			\psfrag{b2}[Bl][Bl][0.59]   {Simulation, \timeout = 0.75}
 			\psfrag{b3}[Bl][Bl][0.59]   {Analysis, \timeout = 0.75}
 			\psfrag{c1}[Bl][Bl][0.59]   {Experiment, \timeout = 2.5}
 			\psfrag{c2}[Bl][Bl][0.59]   {Simulation, \timeout = 2.5}
 			\psfrag{c3}[Bl][Bl][0.59]   {Analysis, \timeout = 2.5}
 			\psfrag{Y1111111111}[bc][bc][0.9] {\ac{AoI} violation probability, \aoivio}
 			\psfrag{X1111111111}[tc][tc][0.9] {Target AoI, $\TargetAoI$}
 			\includegraphics[width=1.0\columnwidth]{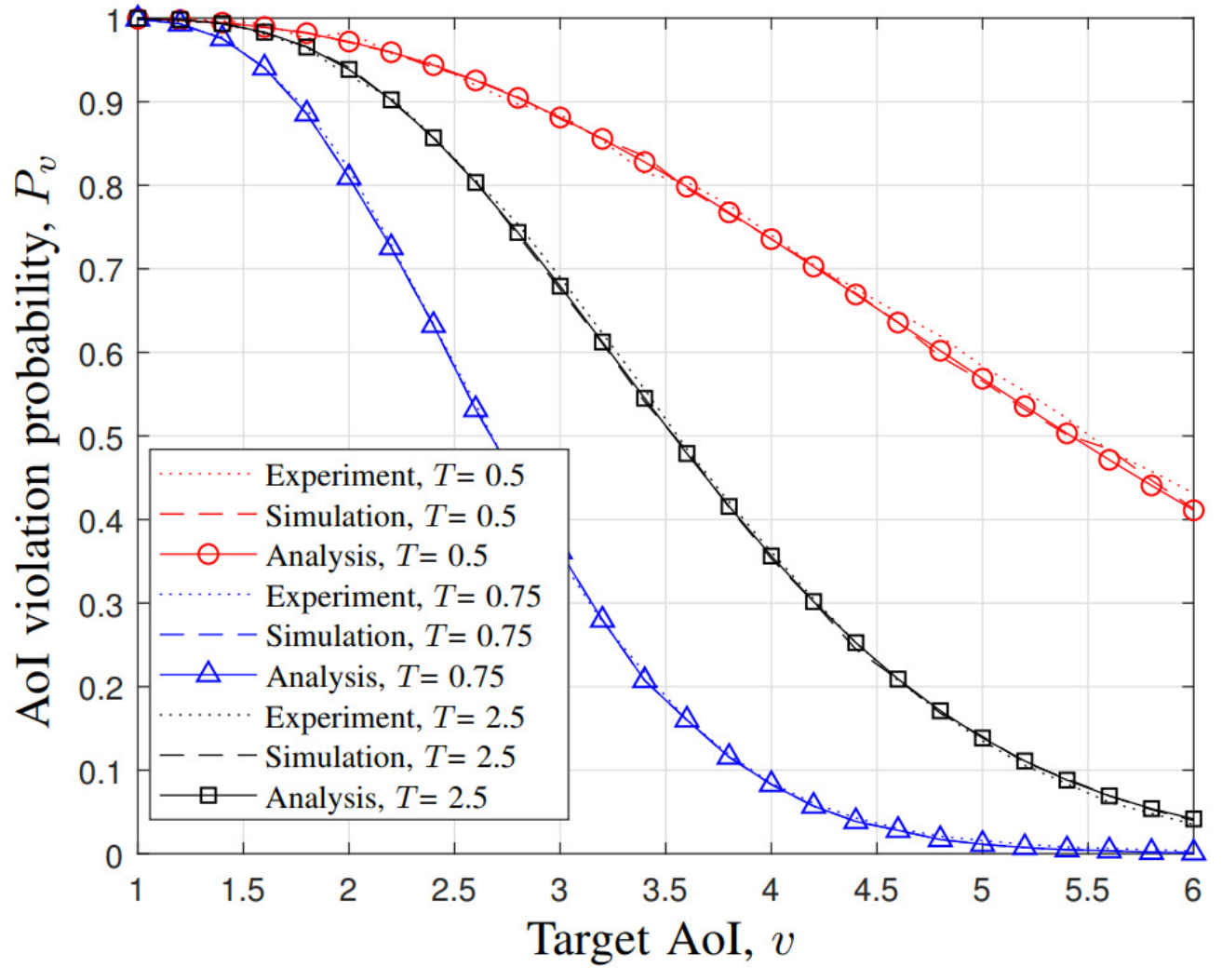}
 		\end{center}
 		
 		\caption{\ac{AoI} violation probability \aoivio{} as a function of target \ac{AoI} $\TargetAoI$ for different timeout \timeout.}
 		
 		\label{fig: timeout_validation}	
 		\vspace{-0.4cm}
 	\end{figure} 
 	In Figs. \ref{fig: STP_validation}, \ref{fig: block_validation}, and \ref{fig: timeout_validation}, we compare the analytical results with the simulation and the experimental results while varying the \ac{BeMN} parameters as a function of $\TargetAoI$. In the simulation results, \aoivio\ is calculated by generating the consensus latencies using the modeled Gamma distribution. In the experimental results, we calculate \aoivio\ by using the measured consensus latencies in the established \ac{HLF} platform. 
 	
 	Figure \ref{fig: STP_validation} presents \aoivio\ for the simulation and the experimental results as a function of $v$ for $\TargetSTP = 0.4, 0.6,$ and $0.8$. We can first observe that the analytical and simulation results show a good match. We also observe that the analytical results match the experimental results, obtained from the \ac{HLF} platform. 
 	Note that for $\TargetSTP = 0.4$, the experimental results also match well with the analytical results although the KS statistic of this case is 0.0623, which is larger than the critical value 0.0515 (see Table \ref{Table 1}). This is because the distribution of the consensus latency and its modeled distribution have similar statistical properties. Specifically, from Table \ref{Table 1}, we can see that the average latency, the average SD, and the skewness of the measured latencies in the experiments are $2.42$, $0.92$, and $0.086$, respectively. These are similar to the values calculated from the modeled distribution, which are $\frac{\alpha}{\beta} = 2.42, \frac{\sqrt{\alpha}}{\beta} = 0.98$, and $\frac{2}{\sqrt{\alpha}}= 0.082$.
 	Therefore, our \ac{AoI} analysis can capture the actual data freshness in \ac{BeMN}.
	In addition, when the target \ac{AoI} $\TargetAoI$ is small (e.g., $\TargetAoI < 3$), the value of the \ac{AoI} violation probability \aoivio\ for  $\TargetSTP = 0.4 $ is less than the case in which $\TargetSTP = 0.8$. This is because the transmission latency \transLat\ increases as $\TargetSTP$ becomes larger, and, thus it becomes difficult to complete the status update within a short time for high $\TargetSTP$. On the other hand, when $\TargetAoI$ is large (e.g., $\TargetAoI \geq 3$), \aoivio\ for  $\TargetSTP = 0.4$ is higher than the case in which $\TargetSTP = 0.8$. In this case, most status updates can be performed without violating the large target \ac{AoI} $v$. Hence, the number of successfully received packets becomes more important than reducing \transLat\ to maintain low \aoivio.
 	
 	Figure \ref{fig: block_validation} shows the analytical, simulation, and experimental results of \aoivio\ as a function of $\TargetAoI$ for $\Blocksize = $ 5, 12, and 25.
  	For small $\TargetAoI$ (e.g., $\TargetAoI < 3$), the value of the \ac{AoI} violation probability \aoivio\ for $\Blocksize = 5 $ is lower than the case in which $\Blocksize = 25$. However, this relationship is reversed when $\TargetAoI$ is large (e.g., $\TargetAoI \geq 3$). This is because the consensus latency distributions of the both cases show different statistical characteristics. From Table \ref{Table 1}, we observe that \aoivio\ for $\Blocksize = 5$ yields larger \ac{SD} and skewness than the case in which $\Blocksize = 25$. Thus, the distribution of $\Blocksize = 5$ is concentrated more on the left, and it has a larger tail probability compared with $\Blocksize = 25$.
 	
  	Figure \ref{fig: timeout_validation} presents the three results of \aoivio\ as a function of $\TargetAoI$ for $\Timeout = 0.5, 0.75,$ and $2.5$.
 	We can see that the \ac{AoI} violation probability \aoivio\ for $\Timeout = 0.5$ achieves the largest mean latency, \ac{SD} and skewness. Although its distribution is highly skewed to the left, it is mostly concentrated near the large mean value. Moreover, the high \ac{SD} results in a large tail probability, which leads to larger \aoivio\ than the other cases.  

	As shown in Figs.~\ref{fig: STP_validation}, \ref{fig: block_validation}, and \ref{fig: timeout_validation}, having a smaller (or larger) value for the target \ac{STP}, block size, or timeout does not always guarantee a lower (or larger) \ac{AoI} violation probability.
	This is due to the conflicting effects of those parameters on the transmission and the consensus latencies. These effects will be discussed more in the following subsection.
 	
 	\begin{figure}[t]
 		
 		\centering
 		\begin{center}
 			\psfrag{Aaaaaaaaaaa}[bc][bc][0.9] {\ac{AoI} and \ac{PAoI} violation probability \aoivio, \paoivio}
 			\psfrag{aaaaaaaaaaa}[tc][tc][0.9] {Target STP, $\TargetSTP$}
 			\psfrag{a}[tc][tc][0.9] {Average \ac{AoI} $\bar{\Delta}$}
 			
 			\psfrag{A111111111111111111111111}[Bl][Bl][0.59]   {\ac{AoI} violation probability, \aoivio{} }
 			\psfrag{b1}[Bl][Bl][0.59]   {\ac{PAoI} violation probability, \paoivio{} }
 			\psfrag{c1}[Bl][Bl][0.59]   {Average \ac{AoI}, \avgaoi{} }
 			
 			\includegraphics[width=1.0\columnwidth]{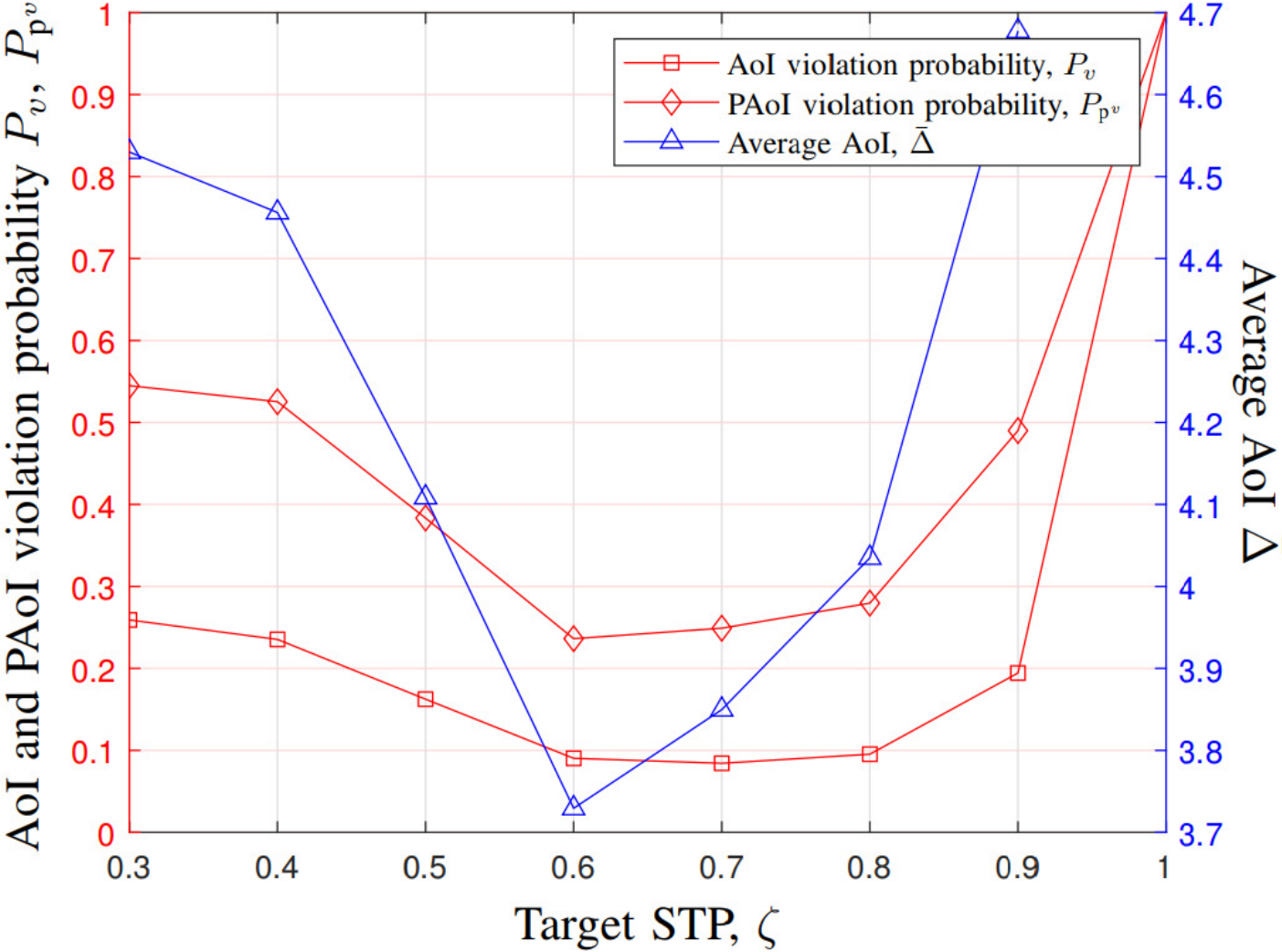}
 		\end{center}
 		
 		\caption{\ac{AoI} violation probability \aoivio{}, \ac{PAoI} violation probability \paoivio{}, and average \ac{AoI} \avgaoi{} as a function of target \ac{STP} $\TargetSTP$ for the target \ac{AoI} $\TargetAoI = 5.5$ and the data size $D = 250$ kb.}
 		
 		\label{fig: STP}	
 		
 	\end{figure}
 	\subsection{Impact of \ac{BeMN} Parameters}
 	Figure \ref{fig: STP} shows \aoivio, \paoivio, and \avgaoi\ for different values of $\TargetSTP$ with $\TargetAoI = 5.5$,  and $D = 250$ kb. From Fig. \ref{fig: STP}, we can observe that \aoivio, \paoivio, and \avgaoi\ have a similar trend according to $\TargetSTP$. In particular, all metrics first decrease and then increase with $\TargetSTP$. This is because when $\TargetSTP$ is small, there are many outages in the packet transmission, so the status is seldom updated, which increases the \ac{AoI}. Hence, in this case, the \ac{AoI}-related performance becomes better as $\TargetSTP$ increases. However, as $\TargetSTP$ increases, the transmission latency \transLat\ also increases. Hence, when $\TargetSTP$ is higher than a certain value, the packets are reliably received at the \ac{HLF}, but the \ac{AoI}-related performance degrades with $\TargetSTP$ due to the longer transmission latency. 
 	
 	\begin{figure}[t]
 		
 		\centering
 		\begin{center}
 			\psfrag{Aaaaaaaaaaa}[bc][bc][0.9] {\ac{AoI} and \ac{PAoI} violation probability \aoivio, \paoivio}
 			\psfrag{aaaaaaaaaaa}[tc][tc][0.9] {Block size, \blocksize}
 			\psfrag{a}[tc][tc][0.9] {Average \ac{AoI}, $\bar{\Delta}$}
 			
 			\psfrag{A111111111111111111111111}[Bl][Bl][0.59] {\ac{AoI} violation probability, \aoivio{} }
 			\psfrag{b1}[Bl][Bl][0.59]   {\ac{PAoI} violation probability, \paoivio{} }
 			\psfrag{c1}[Bl][Bl][0.59]   {Average \ac{AoI}, \avgaoi{} }
 			\includegraphics[width=1.0\columnwidth]{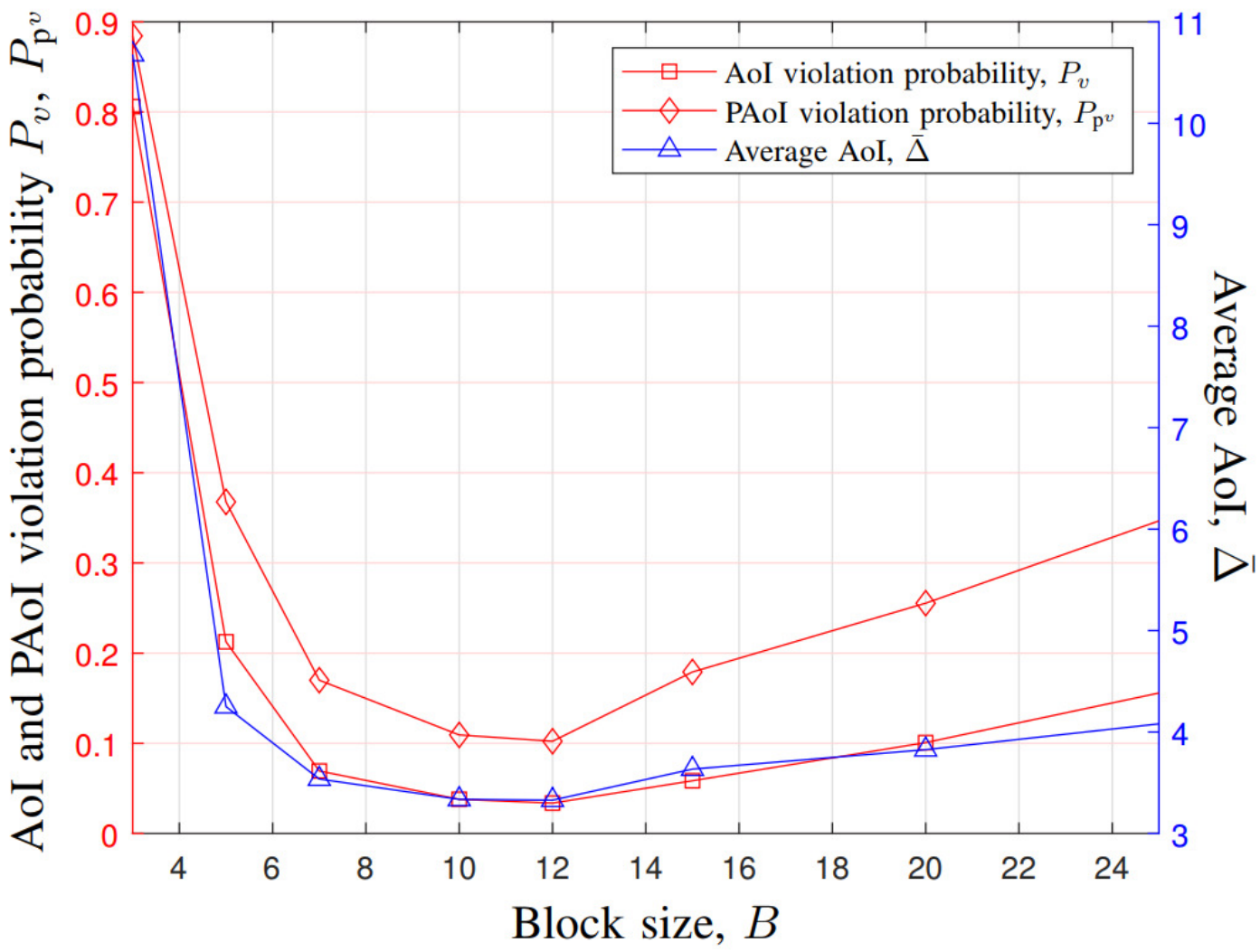}
 		\end{center}
 		
 		\caption{\ac{AoI} violation probability \aoivio{}, \ac{PAoI} violation probability \paoivio{}, and average \ac{AoI} \avgaoi{} as a function of block size \blocksize{} for the target \ac{AoI} $\TargetAoI = 5.5$ and the data size $D = 250$ kb.}
 		
 		\label{fig: blocksize}	
 		
 	\end{figure}
 	
 	Figure \ref{fig: blocksize} presents \aoivio, \paoivio, and \avgaoi\ for different values of \blocksize\ with $\TargetAoI = 5.5$ and $D = 250$ kb. As shown in Fig. \ref{fig: blocksize}, all metrics, \aoivio, \paoivio, and \avgaoi, first decrease and then increase with \blocksize.
 	When \blocksize\ is small, the block generation rate exceeds the block commitment rate in the validation phase, which increases the consensus latency.
 	In particular, the time to commit the $m$ blocks of size $\frac{\Blocksize}{m}$ is always larger than the time to commit the block of size $\Blocksize$ for $m \geq 1$ \cite{PSBV:18}.
 	Hence, in this case, the \ac{AoI}-related performance improves as \blocksize{} increases. 
 	However, as \blocksize{} keeps increasing, the amount of time that a transaction has to wait in a block also increases. Therefore, when \blocksize\ exceeds a certain value, the \ac{AoI}-related performance decreases with \blocksize{} due to longer waiting time in the ordering phase.
 	  	
 	\begin{figure}[t]
 		
 		\centering
 		\begin{center}
 			\psfrag{Aaaaaaaaaaa}[bc][bc][0.9] {\ac{AoI} and \ac{PAoI} violation probability \aoivio, \paoivio}
 			\psfrag{aaaaaaaaaaa}[tc][tc][0.9] {Timeout, \timeout}
 			\psfrag{a}[tc][tc][0.9] {Average \ac{AoI}, $\bar{\Delta}$}
 			
 			\psfrag{A111111111111111111111111}[Bl][Bl][0.59] {\ac{AoI} violation probability, \aoivio{} }
 			\psfrag{b1}[Bl][Bl][0.59]   {\ac{PAoI} violation probability, \paoivio{} }
 			\psfrag{c1}[Bl][Bl][0.59]   {Average \ac{AoI}, \avgaoi{} }
 			\includegraphics[width=1.0\columnwidth]{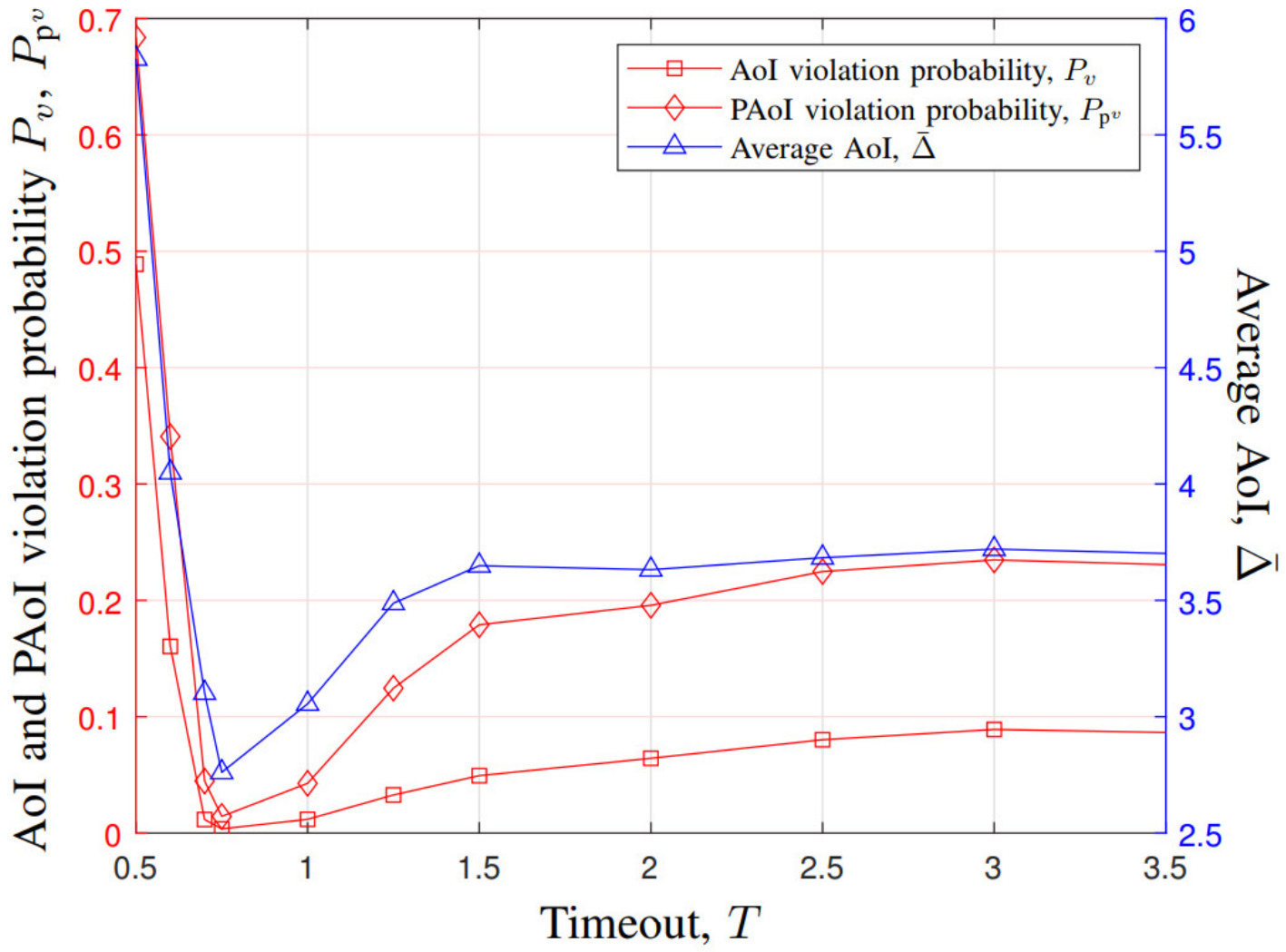}
 		\end{center}
 		
 		\caption{\ac{AoI} violation probability \aoivio{}, \ac{PAoI} violation probability \paoivio{}, and average \ac{AoI} \avgaoi{} as a function of timeout \timeout{} for the target \ac{AoI} $\TargetAoI = 5.5$ and the data size $D = 250$ kb.}
 		
 		\label{fig: timeout}	
 		
 	\end{figure}
 	
 	In Fig. \ref{fig: timeout}, we present \aoivio, \paoivio, and \avgaoi\ for different values of \timeout\ with $\TargetAoI = 5.5$ and $D = 250$ kb. We can see that all metrics first decrease and then increase with \timeout.
 	For small \timeout, the block generation rate in the ordering phase is larger than peers' block commitment rate. Hence, the consensus latency increases, which  increases the \ac{AoI}. In this case, the \ac{AoI}-related performance improves as \timeout\ increases. Nevertheless, as \timeout\ increases, a transaction has to wait longer until the timeout expires. Thus, the \ac{AoI}-related performance degrades with \timeout\ because of the longer latency in the ordering phase. When $\Timeout \times \rho \geq \Blocksize$, most of blocks can be generated before the timeout expires, so the effect of \timeout\ almost vanishes. Therefore, the \ac{AoI}-related performance does not change with \timeout\ any more in this case.

\section{Conclusion} \label{sec: Conclusion}

In this paper, we have studied the statistical characteristics of the \ac{AoI} in \ac{BeMN}. By considering both the transmission latency and the consensus latency, we have presented closed-form expressions for the average \ac{AoI}, the \ac{AoI} violation probability, and the \ac{PAoI} violation probability. We have also obtained an upper and lower bound of the \ac{AoI} violation probability in a simpler form. Further, we have validated our analytical results through the simulations and experiments after constructing the \ac{HLF} platform. We have shown that the analytical results can precisely provide the distribution of the data freshness in \ac{BeMN}.
We have also investigated the effects of \ac{BeMN} parameters on the \ac{AoI} violation probability.
Our results show that: 1)
as the target \ac{STP} increases, the transmission latency becomes longer while the number of reliably received packets for status updates increases, and 2) as the block size or the timeout in an \ac{HLF} network becomes larger, the ordering phase latency increases while the latency in the validation phase decreases due to a smaller load from a lower block generation rate. We have shown that \ac{BeMN} should be properly designed by considering the conflicting effects of the \ac{BeMN} parameters on the \ac{AoI}. The proposed framework can thus be a useful guideline for the design of \ac{BeMN}.

\begin{appendix}

\subsection{Proof of Proposition~\ref{prop 1}} \label{Appendix A}
	From \eqref{channel_capacity} and \eqref{Maximum target rate}, $\bar{\epsilon}_{\Location}$ can be represented as
	\begin{equation}
	\bar{\epsilon}_{\Location} = \max \ 
	\{
	\epsilon \
	|
	\
	\P \left[ \gamma_{x_\text{o}} \geq 2^{\frac{\epsilon}{W}} - 1 \right] \geq \TargetSTP 
	\}.
	\label{Target_STP_2}
	\end{equation}
	From \eqref{SINR_definition}, $\P \left[ \gamma_{x_\text{o}} \geq 2^{\frac{\epsilon}{W}} - 1 \right]$ in \eqref{Target_STP_2} can be given by
	%
	%
	%
	\begin{align}
	\P \hspace{-0.5mm}
	\left[ \hspace{-0.5mm}
	\gamma_{\Location} \hspace{-0.9mm} \geq \hspace{-0.5mm} 2^{\frac{\epsilon}{W}} \hspace{-1.2mm} - \hspace{-0.9mm} 1 
	\right] \hspace{-0.7mm}
	= \hspace{-0.7mm}
	\exp \hspace{-0.9mm}
	\left(
	\hspace{-0.9mm} -\frac{l^n}{P} N_0 W \theta(\epsilon) \hspace{-0.9mm}
	\right) \hspace{-0.9mm}
	\E_{I_{{\Location}}} \hspace{-1.3mm} \left[ 
	\exp \hspace{-0.9mm} 
	\left(\hspace{-0.9mm}
	-\frac{l^n}{P} I_{{\Location}} \theta(\epsilon) \hspace{-0.9mm}  
	\right) \hspace{-0.9mm} 
	\right] \hspace{-0.9mm} ,
	\label{STP_befor_laplace_2}
	\end{align}
	where $\theta(\epsilon) = 2^{\epsilon/W}-1$ by using the \ac{CDF} of the exponential random variable $h_{{\Location}}$.
	
	With the assumption that the distribution of interfering sources follows the \ac{HPPP},\footnote{Note that the distribution of the interfering sources does not follow \ac{HPPP} because their locations are dependent. Nevertheless, this dependency is shown to be weak in \cite{ThHaDhJe:13}.} the Laplace transform of $I_{{\Location}}$ can be given by \cite[Equation 3.21]{Ha:09}
	\begin{align}
	\mathcal{L}_{I_{\Location}} (s) = \exp \left\{ -\lambda \pi s^{2/n} \frac{2 \pi }{n \sin(2 \pi /n)}      \right\}. \label{laplace_interference_1}
	\end{align}
	%
	%
	%
	%
	%
	Using \eqref{laplace_interference_1}, \eqref{STP_befor_laplace_2} can be derived as follows
	\begin{align}
	\P \left[ \gamma_{x_\text{o}} \geq 2^{\frac{\epsilon}{W}} - 1 \right] &= \exp\left( \hspace{-0.5mm} -\frac{l^n}{P} N_0 W \theta({\epsilon}) \hspace{-0.5mm} \right) 
	\nonumber \\
	&\quad \times \exp \left( -\lambda \pi^2 \frac{2l^2 \theta({\epsilon})^{2/n}}{n P^{2/n} \sin \left(2\pi/n  \right)}   \right).
	\label{Target_STP_not_closed}
	\end{align}
	By substituting \eqref{Target_STP_not_closed} into \eqref{Target_STP_2}, $\bar{\epsilon}_{\Location}$ in \eqref{Target_STP_2} can be given by 
	\begin{align}
	\bar{\epsilon}_{\Location} \hspace{-0.5mm} &= \hspace{-0.5mm} \max  \hspace{-0.5mm}
	\left\{ \hspace{-0.5mm}
	\epsilon \ \hspace{-0.5mm} \bigg| 
	\exp \hspace{-0.7mm}
	\left(
	\hspace{-0.9mm} -\frac{l^n}{P} N_0 W \theta({\epsilon}) \hspace{-0.5mm}   
	- \hspace{-0.5mm}
	\frac{-2\lambda \pi^2 l^2 \theta({\epsilon})^{2/n}}{n P^{2/n} \sin \left(2\pi/n  \right)} \hspace{-0.7mm}
	\right)  \hspace{-0.9mm}
	\geq \TargetSTP  \hspace{-0.9mm}
	\right\} \hspace{-0.9mm} . \label{Maximum target STP max}
	\end{align}
	Since $p_c$ is a decreasing function of $\epsilon$, the value of $\bar{\epsilon}_{\Location}$ is the largest $\epsilon$ that satisfies $p_c \geq \TargetSTP$. Hence, $\bar{\epsilon}_{\Location}$ in \eqref{Maximum target STP max} is the same as the one satisfying \eqref{transmissionrate fixed r}.

\subsection{Proof of Theorem~\ref{theorm 1}}\label{app:theorem 1}

	In \eqref{ratio_of_E}, the expectation of $T^\TargetAoI_{k}$ is given by
	\begin{align}
	\E[T^\TargetAoI_{k}] = \integral{0}{\infty} \P[T^\TargetAoI_{k} \geq z] \mathrm{d}z.  \label{expectation_t_k^v}
	\end{align}
	%
	%
	%
	%
	Using \eqref{min_T^v_k} and \eqref{expectation_t_k^v}, $\E[T^\TargetAoI_{k}]$ can be represented by
	\begin{align}
	\E[T^\TargetAoI_{k}]  &= \integral{0}{\Remain} \hspace{-2.0mm}
	\underbrace{ \integral{0}{\infty} \hspace{-2.0mm}
	\P \hspace{-0.5mm} \left[  x \hspace{-0.5mm} + \hspace{-0.5mm} X_{k} \hspace{-0.5mm} + \hspace{-0.5mm} \Tint \hspace{-0.5mm} - \hspace{-0.5mm}  \Remain \hspace{-0.5mm}  \geq \hspace{-0.5mm} z \right] \hspace{-0.5mm} f_{X_{k-1}}(x)  \mathrm{d}z}_{E_1(x)} \mathrm{d}x \nonumber \\
	& \quad + \integral{\Remain}{\infty} \hspace{-1.0mm} \underbrace{ \integral{0}{\infty} \P[X_k + \Tint \geq z] f_{X_{k-1}} (x)\ \mathrm{d}z}_{E_2(x)} \mathrm{d}x, 
	\label{T^v_k} 
	\end{align}
	where $f_{X_{k-1}}(x)$ is in \eqref{gamma_pdf}.
	In \eqref{T^v_k}, $\P \left[  x \hspace{-0.6mm} + \hspace{-0.6mm} X_{k} \hspace{-0.6mm} + \hspace{-0.6mm} \Tint \hspace{-0.6mm} - \hspace{-0.6mm}  \Remain \hspace{-0.9mm} \geq  \hspace{-0.3mm} z \hspace{-0.1mm} \right]$ can be given by
	\begin{align}
	&\P \left[ \Tint  \geq z + \Remain \hspace{-0.5mm} - x - X_k \right] \nonumber \\
	&= \integral{0}{\infty} \P 
	\left[ 
	\Tint \geq z + \Remain - x - w | X_k = w
	\right]  
	f_{X_k} (w) \mathrm{d}w
	\nonumber \\
	& \hspace{-0.9mm} = \hspace{-0.5mm}  \integral{0}{z + \Remain - x} \hspace{-2.0mm}  e^{-\rho \left( z + \Remain -x -w  \right)}  f_{X_k} (w) \mathrm{d}w \nonumber \\
	&\quad + \integral{z + \Remain - x}{\infty} \hspace{-2.0mm}  f_{X_k} (w) \mathrm{d}w \nonumber \\ 
	& \hspace{-0.9mm} = \hspace{-0.9mm}  \frac{\beta^\alpha \hspace{-0.5mm} e^{-\rho \left(z \hspace{-0.3mm} +  \hspace{-0.3mm} \Remain \hspace{-0.3mm} - \hspace{-0.3mm} x  \right)  }}{\Gamma(\alpha)}
	\frac{\gamma \hspace{-0.5mm} 
		\left(
		\hspace{-0.5mm} \alpha, \left( \hspace{-0.5mm} \beta \hspace{-0.5mm} - \hspace{-0.5mm} \rho   \hspace{-0.5mm}     \right) \hspace{-0.5mm}  \left( \hspace{-0.5mm} z \hspace{-0.5mm}  + \hspace{-0.5mm}  \Remain  \hspace{-0.5mm} -   \hspace{-0.5mm}  x  \hspace{-0.5mm} \right)   \hspace{-0.5mm}   
		\right) }
	{(\beta - \rho)^\alpha} \hspace{-0.5mm} \nonumber \\ 
	& \quad + \hspace{-0.5mm} \frac{\Gamma \hspace{-0.7mm}
		\left( 
		\hspace{-0.5mm} \alpha, \beta \hspace{-0.5mm} \left( \hspace{-0.5mm}  z \hspace{-0.5mm}  + \hspace{-0.5mm}  \Remain \hspace{-0.5mm}  - \hspace{-0.5mm}  x  \hspace{-0.5mm} \right) \hspace{-0.5mm}  
		\right)  }
	{\Gamma(\alpha)}, \label{Just Prob}
	\end{align}
	where the second equation is obtained from the fact that $\P \left[ \Tint \geq z + \Remain - x - X_k | X_k \right]$ is always one when $X_k$ is larger than $z + \Remain - x$ and using the exponential distribution of $\Tint$.  In \eqref{Just Prob}, $\gamma(\cdot, \cdot)$ and $\Gamma(\cdot, \cdot)$ are the lower and upper incomplete gamma functions, respectively, i.e.,
	\begin{align}
	\hspace{-0.9mm}	\gamma(\alpha, x) \hspace{-0.5mm} =  \hspace{-0.5mm} \integral{0}{x}  \hspace{-0.7mm} t^{\alpha - 1}  \hspace{-0.5mm} e^{-t} \mathrm{d}t \ \text{and} \ \Gamma(\alpha, x)  \hspace{-0.5mm} =  \hspace{-0.5mm} \integral{x}{\infty}  \hspace{-1.1mm} t^{\alpha - 1}  \hspace{-0.5mm} e^{-t} \mathrm{d}t. \label{definition_of_incomplete}
	\end{align}
	Using \eqref{Just Prob}, $E_1(x)$ in \eqref{T^v_k} can be obtained as
	\begin{align}
	%
	&E_1(x) \hspace{-0.7mm} = \hspace{-0.9mm} \integral{0}{\infty} \hspace{-0.9mm} \left[ \frac{\beta^\alpha e^{-\rho \left(  z + \Remain - x  \right)  } \gamma \left( \alpha, \left(  \beta  -\rho        \right) \hspace{-0.7mm} \left( z + \Remain -x   \right)      \right) }{\Gamma(\alpha) (\beta - \rho)^\alpha} \right. \nonumber \\ 
	&\left. \quad \quad \quad \quad \quad \quad + \frac{\Gamma\left(  \alpha, \beta \left( z + \Remain - x  \right) \right)  }{\Gamma(\alpha)} \right] \mathrm{d}z \nonumber \\
	&\overset{(a)}{=} \hspace{-0.7mm} \frac{\beta^\alpha }{ \Gamma(\alpha)} \hspace{-0.9mm} \summation{n=0}{\infty} \hspace{-0.7mm} \frac{ \left(   \rho \hspace{-0.5mm}  -  \hspace{-0.5mm}  \beta \right)^{n} }   {n! \left(  \alpha  \hspace{-0.5mm}  + \hspace{-0.5mm}  n  \right)} \hspace{-0.7mm} \integral{0}{\infty} \hspace{-4.5mm}   \left(z \hspace{-0.5mm}  +  \hspace{-0.5mm} \Remain \hspace{-0.5mm}   -  \hspace{-0.5mm}  x  \right)^{\alpha + n} e^{-\rho \left( z + \Remain -  x  \right)} \mathrm{d}z \nonumber \\
	& \quad + \frac{1}{\Gamma(\alpha)} \integral{0}{\infty}  \integral{\beta(z + \Remain -x )}{\infty}   t^{\alpha-1} e^{-t} \mathrm{d}t \mathrm{d}z   \nonumber  \\
	&\overset{(b)}{=} \underbrace{\frac{\beta^\alpha}{\Gamma(\alpha) } \summation{n=0}{\infty} \frac{\left(  \rho - \beta \right)^n \Gamma\left( \alpha + n + 1, \rho \left(\Remain - x \right)   \right) }{n! \left( \alpha + n \right) \rho^{\alpha +n +1}} }_{f_1(x)} \nonumber \\ 
	& \quad + \underbrace{\frac{ \Gamma \left( \alpha \hspace{-0.5mm} + \hspace{-0.5mm} 1, \beta \left( \Remain \hspace{-0.5mm} - \hspace{-0.5mm} x \right)    \right) }  {\beta \Gamma(\alpha)}}_{f_2(x)} \hspace{-0.7mm} + \hspace{-0.7mm} \underbrace{\frac{ \left( x \hspace{-0.5mm} - \hspace{-0.5mm} \Remain \right) }{\Gamma(\alpha) } \Gamma\left(   \alpha, \beta \left( \Remain \hspace{-0.5mm} - \hspace{-0.5mm} x  \right)   \right)}_{f_3(x)},  \label{substitution}
	\end{align}
	where $(a)$ follows from the fact that the lower incomplete Gamma function $\gamma(\alpha, x)$ can be represented as \cite[Equation 8.354-1]{ToI}
	\begin{align}
	\gamma(\alpha, x) = \summation{n=0}{\infty} \frac{(-1)^n x^{\alpha + n}}{n! (\alpha + n)},  \label{8.354-1}
	\end{align}
	and (b) is obtained by changing the order of the integral and \eqref{definition_of_incomplete}. Using \eqref{substitution}, the integral of $E_1(x)$ in \eqref{T^v_k} can be represented as
	\begin{align}
	&\integral{0}{\Remain} \hspace{-2.0mm} E_1(x) \mathrm{d}x \hspace{-0.4mm} = \hspace{-0.7mm} \integral{0}{\Remain} \hspace{-0.7mm} \left\{f_1(x) \hspace{-0.7mm} + \hspace{-0.7mm} f_2(x) \hspace{-0.7mm}  + \hspace{-0.7mm} f_3(x) \right\} f_{X_{k-1}}(x) \mathrm{d}x , \label{f_substitution}
	\end{align}
	where $f_i(x)$ are defined in \eqref{substitution} for $i \in \{1, 2, 3\}$. In \eqref{f_substitution}, the first term can be given by 
	\begin{align}
	&\integral{0}{\Remain} \hspace{-1.5mm}  f_1(x) f_{X_{k-1}}(x) \mathrm{d}x \nonumber \\
	&\overset{(a)} {=} \frac{\beta^{2\alpha}}{\Gamma(\alpha)^2 } \summation{n=0}{\infty} \frac{\left(  \rho - \beta \right)^n}{n! \left( \alpha + n \right) \rho^{\alpha +n +1}} \integral{0}{\Remain} \bigg[  \Gamma\left( \alpha +n + 1 \right)   \nonumber \\
	& \quad  - \summation{k=0}{\infty} \frac{(-1)^k \left\{ \rho \left ( \Remain -x \right)  \right\}^{\alpha+n+k+1} }{k! \left( \alpha +n+k+1 \right)}  \bigg ] x^{\alpha - 1} e^{-\beta x} \mathrm{d}x  \nonumber \\
	%
	& \overset{(b)} {=} \hspace{-0.5mm} \frac{\beta^{2\alpha}}{\Gamma(\alpha)^2 } \hspace{-0.7mm} \summation{n=0}{\infty} \hspace{-0.5mm} \frac{\left(  \rho - \beta \right)^n}{n! \left( \alpha + n \right) \rho^{\alpha +n +1}} \hspace{-0.5mm} \bigg[ \frac{ \Gamma\left( \alpha \hspace{-0.5mm} + \hspace{-0.5mm} n \hspace{-0.5mm}  + \hspace{-0.5mm}  1 \right)} {\beta^\alpha} \gamma\left(\alpha, \beta \Remain \right)  \nonumber \\
	& \quad -\summation{k=0}{\infty} \frac{(-1)^k (\rho \Remain)^{\alpha+n+k+1} }{k! \left( \alpha +n+k+1 \right)}  \Remain^{\alpha} B\left(\alpha \hspace{-0.5mm} + \hspace{-0.5mm} n \hspace{-0.5mm} + \hspace{-0.5mm} k \hspace{-0.5mm} + \hspace{-0.5mm} 2, \alpha \right)    \nonumber \\
	& \quad \quad \quad \quad \times {_1}F_1 \left( \alpha; 2\alpha  \hspace{-0.5mm} + \hspace{-0.5mm}  n  \hspace{-0.5mm} +  \hspace{-0.5mm} k  \hspace{-0.5mm}  + \hspace{-0.5mm}  2; \beta \Remain \right) \bigg],  \label{1_3}
	\end{align}
	%
	where $B(\cdot, \cdot)$ is the beta function, and ${_1}F_1 (\cdot \ ; \cdot \ ; \cdot)$ is the confluent hypergeometric function. In \eqref{1_3}, $(a)$ is obtained since $\Gamma(\alpha, x) = \Gamma(\alpha) - \gamma(\alpha, x),$ and $(b)$ follows from \cite[Equation 3.383-1]{ToI}. 
	Similarly, in \eqref{f_substitution}, the second term can be obtained as
	\begin{align}
	&\integral{0}{\Remain}  f_2(x) f_{X_{k-1}}(x)\mathrm{d}x \nonumber \\
	&=\frac{\beta^{2\alpha}}{\Gamma(\alpha)^2} \hspace{-0.7mm}
	\integral{0}{\Remain} \hspace{-0.9mm}
	\left[ \hspace{-0.0mm} \Gamma  ( \alpha + \hspace{-0.0mm}  1 )  - \hspace{-0.0mm} \summation{n=0}{\infty} \hspace{-0.0mm} \frac{(-\beta)^n \hspace{-0.0mm} \left ( \Remain \hspace{-0.0mm}  -  \hspace{-0.0mm} x \right)^{\alpha \hspace{-0.0mm} + \hspace{-0.0mm} n \hspace{-0.0mm} + \hspace{-0.0mm} 1} }{n! \left( \alpha  \hspace{-0.0mm} +  \hspace{-0.0mm}n \hspace{-0.0mm} + \hspace{-0.0mm} 1 \right)} \right] \nonumber \\
	& \quad \quad \quad \quad \quad \times \hspace{-0.0mm} x^{\alpha \hspace{-0.0mm} - \hspace{-0.0mm} 1}  \hspace{-0.0mm} e^{-\beta x} \hspace{-0.0mm} \mathrm{d}x \nonumber \\
	&= \hspace{-0.5mm} \frac{\alpha \gamma(\alpha, \beta \Remain)}
	{\beta \Gamma(\alpha)} \hspace{-0.5mm}
	 -\hspace{-0.5mm}
	  \summation{n=0}{\infty} \hspace{-0.7mm}
	  \frac{(-1)^n (\beta \Remain)^{2\alpha+n} }
	  {n! \left( \alpha +n+1 \right)}
	  B(\alpha \hspace{-0.5mm} +  \hspace{-0.5mm}  n \hspace{-0.5mm} +  \hspace{-0.5mm} 2, \alpha) \nonumber \\
	&\quad \quad \quad \quad \quad \quad \quad \quad \quad \quad  \hspace{-1.2mm} \times \hspace{-0.7mm} \frac{\Remain}
	{\Gamma(\alpha)^2}
	{_1}F_1(\alpha; 2\alpha \hspace{-0.5mm} +  \hspace{-0.5mm} n  \hspace{-0.5mm}+ \hspace{-0.5mm} 2 ; -\beta \Remain) 
	 . \label{2_2}
	\end{align}   
	The third term in \eqref{f_substitution} is given by
	\begin{align}
	&\integral{0}{\Remain} f_3(x) f_{X_{k-1}}(x) \mathrm{d}x \nonumber \\
	%
	&\overset{(a)}{=} -\frac{\beta^\alpha}{\Gamma(\alpha)^2} \hspace{-0.5mm}
	\integral{0}{\Remain} \hspace{-0.9mm}
	k \Gamma(\alpha, \beta k) (\Remain -k)^{\alpha - 1} e^{-\beta (\Remain - k)} \mathrm{d}k \nonumber \\
	&\overset{(b)}{=} -\frac{\beta^\alpha}{\Gamma(\alpha)^2} \integral{0}{\Remain} \left\{ \Gamma(\alpha)  - \summation{n=0}{\infty} \frac{(-1)^n (\beta k)^{\alpha +n}  }{n! (\alpha +n)}  \right\} \nonumber \\
	& \quad \times 
	k (\Remain - k)^{\alpha - 1} e^{-\beta (\Remain - k)} \mathrm{d}k \nonumber\\
	&\overset{(c)}{=} -\frac{\Remain(\beta \Remain)^\alpha e^{-\beta \Remain} }{\Gamma(\alpha)} B(\alpha, 2)  {_1}F_1(2; \alpha \hspace{-0.5mm} + \hspace{-0.5mm} 2; \beta \Remain) \nonumber \\
	& \quad + \frac{\Remain e^{-\beta \Remain}}{\Gamma(\alpha)^2} \summation{n=0}{\infty} \frac{(-1)^n (\beta \Remain)^{2\alpha +n}  }{n! (\alpha +n)} B(\alpha, \alpha \hspace{-0.5mm} +  \hspace{-0.5mm} n \hspace{-0.5mm}  + \hspace{-0.5mm} 2) \nonumber \\
	& \quad \quad \quad \quad \quad \quad \quad \quad \times 
	{_1}F_1(\alpha \hspace{-0.5mm}+\hspace{-0.5mm} n \hspace{-0.5mm}+\hspace{-0.5mm} 2; 2\alpha  \hspace{-0.5mm} + \hspace{-0.5mm} n \hspace{-0.5mm} + \hspace{-0.5mm} 2; \beta \Remain) \nonumber \\
	&\overset{(d)}{=} -\frac{\Remain(\beta \Remain)^\alpha }{\Gamma(\alpha)} B(\alpha, 2)  {_1}F_1(\alpha; \alpha \hspace{-0.5mm} + \hspace{-0.5mm} 2; -\beta \Remain) \nonumber \\
	& \quad + \frac{\Remain}{\Gamma(\alpha)^2} \summation{n=0}{\infty} \frac{(-1)^n (\beta \Remain)^{2\alpha +n}  }{n! (\alpha +n)} B(\alpha, \alpha \hspace{-0.5mm} +  \hspace{-0.5mm} n \hspace{-0.5mm}  + \hspace{-0.5mm} 2) \nonumber \\
	& \quad \quad \quad \quad \quad \quad \quad \times {_1}F_1(\alpha; 2\alpha  \hspace{-0.5mm} + \hspace{-0.5mm} n \hspace{-0.5mm} + \hspace{-0.5mm} 2; -\beta \Remain), \label{3_3}
	\end{align}
	where $(a)$ is obtained by substituting $k$ for $\Remain - x$, $(b)$ follows from \eqref{8.354-1}, $(c)$ is obtained by the similar steps, which are used in \eqref{1_3}, and $(d)$ stems from ${_1}F_1 (a;b;z) = e^z {_1}F_1(b-a;b;-z)$ in \cite[Equation 9.212-1]{ToI}.  

	Now, we derive the integral of $E_2(x)$ in \eqref{T^v_k}. In the integral range of $z$, $X_k + \Tint \geq 0$ holds. Hence, we have
	\begin{align}
	&\integral{\Remain}{\infty} E_2(x) \mathrm{d}x =  \integral{\Remain}{\infty} \hspace{-0.5mm} \integral{0}{\infty} \hspace{-0.5mm}
	 \P[X_k + \Tint \geq z]  \mathrm{d}z f_{X_{k-1}} (x) \mathrm{d}x \nonumber  \\ 
	&\overset{(a)}{=} \left( \E[X_k] + \E [\Tint] \right) \integral{\Remain}{\infty} f_{X_{k-1}} (x) \mathrm{d}x \nonumber \\ 
	&= \left( \frac{\alpha}{\beta} + \frac{1}{\rho}  \right) \frac{\Gamma(\alpha, \beta \Remain)}{\Gamma(\alpha)} \label{T_ak_2},
	\end{align}
	where $(a)$ is obtained since $X_k$ and $\Tint$ are independent.
	
	Finally, $\E[T^\TargetAoI_{k}]$ is the summation of \eqref{1_3}, \eqref{2_2}, \eqref{3_3}, and \eqref{T_ak_2}.
	Note that $ \displaystyle  \E[T_k] = \E[X_k + \Tint] = \frac{\alpha}{\beta} + \frac{1}{\rho}$. Therefore, we obtain \aoivio{} as the ratio of $\E[T^\TargetAoI_{k}]$ and $\E[T_k]$ as \eqref{AoI_violation_Prob}.   

\end{appendix}

\bibliographystyle{IEEEtran}
\bibliography{StringDefinitions,IEEEabrv,mybib}

\end{document}